\documentclass[10pt,journal]{IEEEtran}

\usepackage{amssymb}
\usepackage{amsmath}
\usepackage{cite}
\usepackage{url}
\usepackage{xcolor}
\usepackage{cite,graphicx,amsmath,amssymb}
\usepackage{subfigure}
\usepackage{fancyhdr}
\usepackage{mdwmath}
\usepackage{mdwtab}
\usepackage{caption}
\usepackage{amsthm}
\usepackage{setspace}
\usepackage{hyperref}
\usepackage{algorithm}
\usepackage{algorithmic}
\usepackage{multirow}
\usepackage{makecell}
\usepackage{mathtools}
\usepackage{mathrsfs}
\usepackage{bm}
\usepackage{pifont}
\hypersetup{colorlinks=false}

\graphicspath{{./src/img/}}

\DeclareMathOperator*{\argmax}{arg\,max}
\newtheorem{remark}{Remark}
\newtheorem{theorem}{Theorem}

\newtheorem{lemma}{Lemma}

\newtheorem{corollary}{Corollary}

\newtheorem{proposition}{Proposition}

\allowdisplaybreaks
\expandafter\def\expandafter\normalsize\expandafter{%
    \normalsize%
    \setlength\abovedisplayskip{5pt}%
    \setlength\belowdisplayskip{5pt}%
    \setlength\abovedisplayshortskip{2pt}%
    \setlength\belowdisplayshortskip{2pt}%
}

\title{Beamfocusing Optimization for Near-Field Wideband Multi-User Communications}

\author{
        Zhaolin~Wang,~\IEEEmembership{Graduate Student Member,~IEEE,}
        Xidong~Mu,~\IEEEmembership{Member,~IEEE,} \\
        and Yuanwei~Liu,~\IEEEmembership{Follow,~IEEE}
        
\thanks{An earlier version of this paper was presented in part at the IEEE Global Communications Conference, Kuala Lumpur, Malaysia, 2023 \cite{conference_version}.
        The authors are with the School of Electronic Engineering
and Computer Science, Queen Mary University of London, London E1 4NS, U.K. (e-mail: \{zhaolin.wang, xidong.mu, yuanwei.liu\}@qmul.ac.uk).}
\vspace{-0.3cm}
}

\begin{document}

\maketitle
\begin{abstract}
    
    A near-field wideband communication system is investigated in which a base station (BS) employs an extra-large scale antenna array (ELAA) to serve multiple users in its near-field region. To facilitate near-field multi-user beamforming and mitigate the spatial wideband effect, the BS employs a hybrid beamforming architecture based on true-time delayers (TTDs). In addition to the conventional fully-connected TTD-based hybrid beamforming architecture, a new sub-connected architecture is proposed to improve energy efficiency and reduce hardware requirements. Two wideband beamforming optimization approaches are proposed to maximize spectral efficiency for both architectures. 1) Fully-digital approximation (FDA) approach: In this method, the TTD-based hybrid beamformer is optimized by the block-coordinate descent and penalty method to approximate the optimal digital beamformer. This approach ensures convergence to the stationary point of the spectral efficiency maximization problem. 2) Heuristic two-stage (HTS) approach: In this approach, the analog and digital beamformers are designed in two stages. In particular, two low-complexity methods are proposed to design the high-dimensional analog beamformers based on approximate and exact line-of-sight channels, respectively. Subsequently, the low-dimensional digital beamformer is optimized based on the low-dimensional equivalent channels, resulting in reduced computational complexity and channel estimation complexity. Our numerical results show that 1) the proposed approach effectively eliminates the spatial wideband effect, and 2) the proposed sub-connected architecture is more energy efficient and has fewer hardware constraints on the TTD and system bandwidth compared to the fully-connected architecture.
    
\end{abstract}



\section{Introduction} \label{sec:intro}
The sixth-generation (6G) wireless network, with its substantially higher performance targets than current networks, is expected to integrate two key emerging trends: extremely large-scale antenna arrays (ELAAs) and the use of significantly higher frequency bands \cite{zhang20196g}. ELAAs can offer substantial array gain and high spatial resolution, thereby enhancing network capacity and connectivity. Simultaneously, higher frequency bands, such as millimeter-wave (mmWave) and terahertz (THz) bands, are anticipated to provide vast bandwidth resources, further advancing network performance. However, the shift towards ELAAs and the adoption of mmWave or THz frequency bands involves more than just an increase in the number of antennas and carrier frequencies. It also leads to notable changes in the electromagnetic properties of the wireless environment. The electromagnetic field region around transmit antennas is typically divided into two regions: the near-field region and the far-field region \cite{kraus2002antennas}. The Rayleigh distance, which is proportional to the antenna array's aperture and the carrier frequency, is a widely used measure to distinguish between these two regions. In previous generations of wireless networks, the Rayleigh distance was generally less than a few meters, rendering the near-field region relatively insignificant. This allowed for wireless network designs based predominantly on the far-field approximation. However, with the emergence of ELAAs and the use of higher frequency bands, the Rayleigh distance has extended considerably, spanning tens or even hundreds of meters \cite{liu2023near}. This necessitates a fundamental rethinking and redesign of wireless network architectures by taking into account the near-field effects.

Beamforming in near-field systems is different from the conventional far-field counterpart. In far-field systems, the signal wavefront is approximated to be planar. Therefore, far-field beamforming can only steer signals along a direction like a flashlight \cite{heath2016overview}. However, near-field systems require a more accurate spherical wavefront model. In this case, near-field beamforming can focus signals around a specific location. Thus, near-field beamforming is also referred to as \emph{beamfocusing} \cite{zhang2022beam}. Compared to far-field beamforming, near-field beamfocusing can enhance the data rate and multiple access capability by reducing inter-user interference. Phased arrays are a common hardware choice for beamforming, typically employing phase shifters (PSs). These arrays offer cost efficiency and high energy efficiency compared to fully-digital arrays. By utilizing a limited number of radio-frequency (RF) chains, phased arrays can achieve performance comparable to fully-digital arrays through hybrid analog-digital beamforming in narrowband systems \cite{yu2016alternating, sohrabi2016hybrid}. However, traditional phased arrays face limitations when applied to ELAAs, particularly in exploiting the wide bandwidth resources available at mmWave and THz bands. This limitation arises from the spatial wideband effect \cite{wang2018spatial}, where the array response varies significantly with frequency but the conventional phased array can only facilitate frequency-independent beamforming. In far-field wideband systems, the mismatch between array response and beamforming ability leads to the misalignment of the beams with respect to the user's direction at different frequencies. The case becomes even worse in near-field wideband systems, wherein the beams can misfocus in both range and angle domains with respect to the user's location. Therefore, the frequency-independent property of conventional phased arrays limits the effective usage of wideband resources and leads to low spectral efficiency. The main focus of this paper is to address this issue in near-field wideband multi-user systems.

\begin{table*}[t!]
    \caption{Our Contributions in Contrast to the State-of-the-Art}
    \label{table_compare}
    \centering
    \begin{tabular}{!{\vrule width1pt}l!{\vrule width1pt}c!{\vrule width1pt}c!{\vrule width1pt}c!{\vrule width1pt}c!{\vrule width1pt}c!{\vrule width1pt}c!{\vrule width1pt}}
        \Xhline{1pt}
        &[6], [10]-[12] & [13] &[14], [16]-[19] & [20] & [21], [22] & \textbf{This work} \\ \Xhline{1pt}
        Near-field effect &\ding{52} &\ding{56} &\ding{56} &\ding{52} &\ding{52} &\ding{52}\\ \Xhline{1pt}
        Spatial wideband effect &\ding{56} &\ding{52} &\ding{52} &\ding{52} &\ding{52} &\ding{52}\\ \Xhline{1pt}
        Multiple communication users &\ding{52} &\ding{56} &\ding{56} &\ding{56} &\ding{56} &\ding{52}\\ \Xhline{1pt}
        The employment of TTDs &\ding{56} &\ding{56} &\ding{52} &\ding{56} &\ding{52} &\ding{52}\\\Xhline{1pt}      
    \end{tabular}
    \vspace{-0.3cm}
\end{table*}

\subsection{Prior Works}

The near-field narrowband beamforming has been widely explored in the literature. For example, \cite{lu2021near} studied the performance of near-field narrowband multi-user communications using the classical maximum-ratio transmission (MRT), zero-forcing (ZF), and minimum mean-squared error (MMSE) beamforming strategies. Their results unveiled the advantages of near-field beamfocusing over conventional far-field beamforming in terms of inter-user interference mitigation. From the optimization perspective, \cite{zhang2022beam} proposed a series of alternating optimization approaches for near-field narrowband beamfocusing using different antenna arrays. Focusing on the holographic metasurface antenna arrays, \cite{10339299} developed multiple multi-user nearfield beamforming designs under different hardware constraints. Furthermore, by analyzing the channel gain and interference gain, \cite{bacci2023spherical} reveals the significant performance loss caused by utilizing the far-field approximation. Their findings emphasize the urgent need for a transition to near-field communication in the high-frequency band.

In parallel, wideband beamforming, which aims to mitigate the spatial wideband effect, has also been well-studied for far-field systems. In \cite{chen2020hybrid}, several matrix projection methods were introduced to alleviate the spatial wideband effect for conventional phased arrays. In a separate study, \cite{gao2021wideband} proposed a virtual sub-array approach to reduce this effect. However, due to the frequency-independent nature of phased arrays, these methods still exhibit significant performance loss compared to the fully-digital array for large bandwidth. To tackle this issue, true-time delayers (TTDs) can be employed in phase arrays to facilitate the frequency-dependent beamforming \cite{6531062}. However, compared to PSs, TTDs are typically more expensive, less energy-efficient, and larger in size. Several TTD-based hybrid beamforming architectures have been proposed to reduce the negative impacts of TTDs. For example, the authors of \cite{gao2021wideband} and \cite{dai2022delay} crafted a hybrid-TTD-PS hybrid beamformer architecture that has massive PSs but only a limited number of TTDs, and explored its performance in a far-field single-user system. \cite{nguyen2022joint} studied the optimization of the PS and TTD coefficients in a far-field single-user system, and derived a closed-form solution for limited-range TTDs based on the far-field approximation. The channel estimation problem was investigated in \cite{dovelos2021channel}, developing a modified orthogonal matching pursuit (OMP) algorithm. To further reduce energy consumption and hardware complexity, \cite{yan2022energy} exploited the low-cost TTDs with fixed time delays and explored the related beamforming design for far-field single-user systems. 

So far, there are only a limited number of studies about the near-field wideband beamforming design \cite{myers2021infocus, cui2021near, zhang2023deep}. Specifically, the authors of \cite{myers2021infocus} developed a beamforming design that utilizes spatial coding to address the near-field beam split effect for the conventional hybrid beamformer architecture. Adopting the TTD-based hybrid beamforming architecture, a heuristic single-user wideband beamforming design was proposed in \cite{cui2021near} based on piecewise-far-field approximation. As a further advance, a deep-learning-based beamforming method was conceived in \cite{zhang2023deep} to maximize the array gain at a single user across all subcarriers. 

\vspace{-0.1cm}
\subsection{Motivations and Contributions}

Against the above background, the TTD-based hybrid beamforming architecture is a promising solution for the wideband beamforming design in either far-field or near-field systems.  However, it is notable that most existing research in this area has predominantly focused on single-user systems \cite{gao2021wideband, dai2022delay, nguyen2022joint,zhang2022fast, dovelos2021channel, yan2022energy, cui2021near, zhang2023deep}. In such contexts, beamforming design is typically oriented towards maximizing array gain along the direction or location of the user. This approach, while effective for single-user scenarios, is flawed in multi-user systems where both array gain and inter-user interference must be carefully balanced. This requirement introduces a complex design challenge, particularly in the context of TTD-based hybrid beamforming. To the best of the authors' knowledge, this is the first attempt to directly address the design problem of TTD-based hybrid beamforming in multiuser systems by considering both unit modulus constraints and finite TTD range constraints in either far-field or near-field scenarios. Furthermore, the conventional TTD-based hybrid beamforming architecture employs a fully-connected structure. In such a structure, the number of TTDs and their required maximum delay are proportional to the aperture of the entire antenna array \cite{dai2022delay, cui2021near}. Therefore, for ELAAs with large apertures, the requirement on TTDs can be stringent and costly. Motivated by this, we proposed a new sub-connected TTD-based architecture that not only alleviates the stringent TTD requirement but also enhances energy efficiency. Finally, considering that the far-field model is essentially an approximation of the near-field model \cite{liu2023near}, our study primarily focuses on beamforming design in near-field systems. Our design can be directly extended to far-field systems. The primary contributions of this paper can be succinctly summarized below, which are boldly and explicitly contrasted to the relevant state-of-the-art in Table \ref{table_compare}.
\begin{itemize}
    \item We first investigated near-field beamforming designs that maximize spectral efficiency in wideband multi-user systems employing TTD-based hybrid beamforming architectures. To improve energy efficiency and reduce hardware requirements, we propose a new sub-connected TTD-based hybrid beamforming architecture.
    \item We propose a penalty-based fully-digital approximation (FDA) method for near-field beamforming designs in both fully-connected and sub-connected architectures. In this approach, the optimization variables are iteratively optimized by closed-form solutions or low-complexity one-dimensional search. This approach also guarantees convergence to the stationary point of the spectral efficiency maximization problem.
    \item We then propose a low-complexity heuristic two-stage (HTS) near-field beamforming design method. In particular, the analog beamformer is first designed to maximize the line-of-sight (LoS) signal power for each user. We propose two low-complexity methods to design the analog beamformers based on the approximated and exact LoS channels, respectively. Subsequently, the digital beamformer is optimized to maximize spectral efficiency through the low-dimensional equivalent channel.
    \item Numerical results demonstrate the efficiency of the proposed method, revealing that 1) the spatial wideband effect can be effectively eliminated by utilizing TTD; 2) the penalty-based FDA method yields the best performance in both fully-connected and sub-connected architectures; 3) the HTS method can achieve comparable performance to the penalty-based FDA method in most cases; and 4) Compared to fully-connected architectures, sub-connected architectures demonstrate superior energy efficiency and enhanced robustness in the face of ultra-large bandwidths and low TTD maximum delays.
\end{itemize}

\subsection{Organization and Notations}
The remainder of this paper is organized as follows. Section \ref{sec:model} presents the signal model for fully-connected and sub-connected architectures and formulates the spectral efficiency maximization problem. Sections \ref{sec:FDA} and \ref{sec:HTS} present the details of the proposed FDA and HTS approaches, respectively. Section \ref{sec:result} demonstrates the numerical results. Finally, this paper is concluded in Section \ref{sec:conclusion}.

\emph{Notations:} Scalars, vectors, and matrices are represented by the lower-case, bold-face lower-case, and bold-face upper-case letters, respectively; The transpose, conjugate transpose, pseudo-inverse, and trace of a matrix are denoted by $\mathbf{X}^T$, $\mathbf{X}^H$, $\mathbf{X}^\dagger$, and $\mathrm{tr}(\mathbf{X})$, respectively. The Euclidean norm of vector $\mathbf{x}$ is denoted as $\|\mathbf{x}\|$, while the Frobenius norm of matrix $\mathbf{X}$ is denoted as $\|\mathbf{X}\|_F$. For matrix $\mathbf{X}$, $\mathbf{X}(i,j)$ refers to its entry in the $i$-th row and $j$-th column. $\mathbf{X}(i:j,:)$ represents a matrix composed of the rows from the $i$-th to the $j$-th, and $\mathbf{X}(:,i)$ represents a vector composed of the $i$-th column. $\mathbf{x}(i)$ indicates the $i$-th entry of vector $\mathbf{x}$. A block diagonal matrix with diagonal blocks $\mathbf{x}_1,\dots,\mathbf{x}_N$ is denoted as $\mathrm{blkdiag}(\mathbf{x}_1,\dots,\mathbf{x}_N)$. $\mathbb{E}[\cdot]$ represents the statistical expectation, while $\Re\{\cdot\}$ refers to the real component of a complex number. $\mathcal{CN}(\mu, \sigma^2)$ denotes the circularly symmetric complex Gaussian random distribution with mean $\mu$ and variance $\sigma^2$. The matrix $\mathbf{I}_N$ denotes the identity matrix of size $N$. The symbol $\angle$ represents the phase of a complex value.
\section{System Model and Problem Formulation} \label{sec:model}

We consider a THz near-field wideband multi-user communication system in which a BS utilizes a uniform linear array (ULA) consisting of $N$ antenna elements with an antenna spacing of $d$. The system involves $K$ communication users, each equipped with a single antenna, and their indices are collected in the set $\mathcal{K}$. To mitigate inter-symbol interference inherent in wideband communications, orthogonal frequency division multiplexing (OFDM) is employed, where the signal is generated in the frequency domain and then transformed into the time domain by inverse discrete Fourier transform (IDFT). Let $M$ and $L_{\mathrm{CP}}$ denote the number of subcarriers and the length of cyclic prefix in OFDM, respectively, $B$ denote the system bandwidth, $f_c$ denote the central carrier frequency, $c$ denote the speed of light, and $\lambda_c = c/f_c$ denote the wavelength of the central carrier. Consequently, the frequency of subcarrier $m$ is determined as $f_m = f_c + B(2m-1-M)/(2M), \forall m \in \mathcal{M} = \{1,...,M\}$. The antenna spacing of the ULA at the BS is set to half of the central wavelength, i.e., $d = \lambda_c/2$.  

\begin{figure}[t!]
    \centering
    \includegraphics[width=0.3\textwidth]{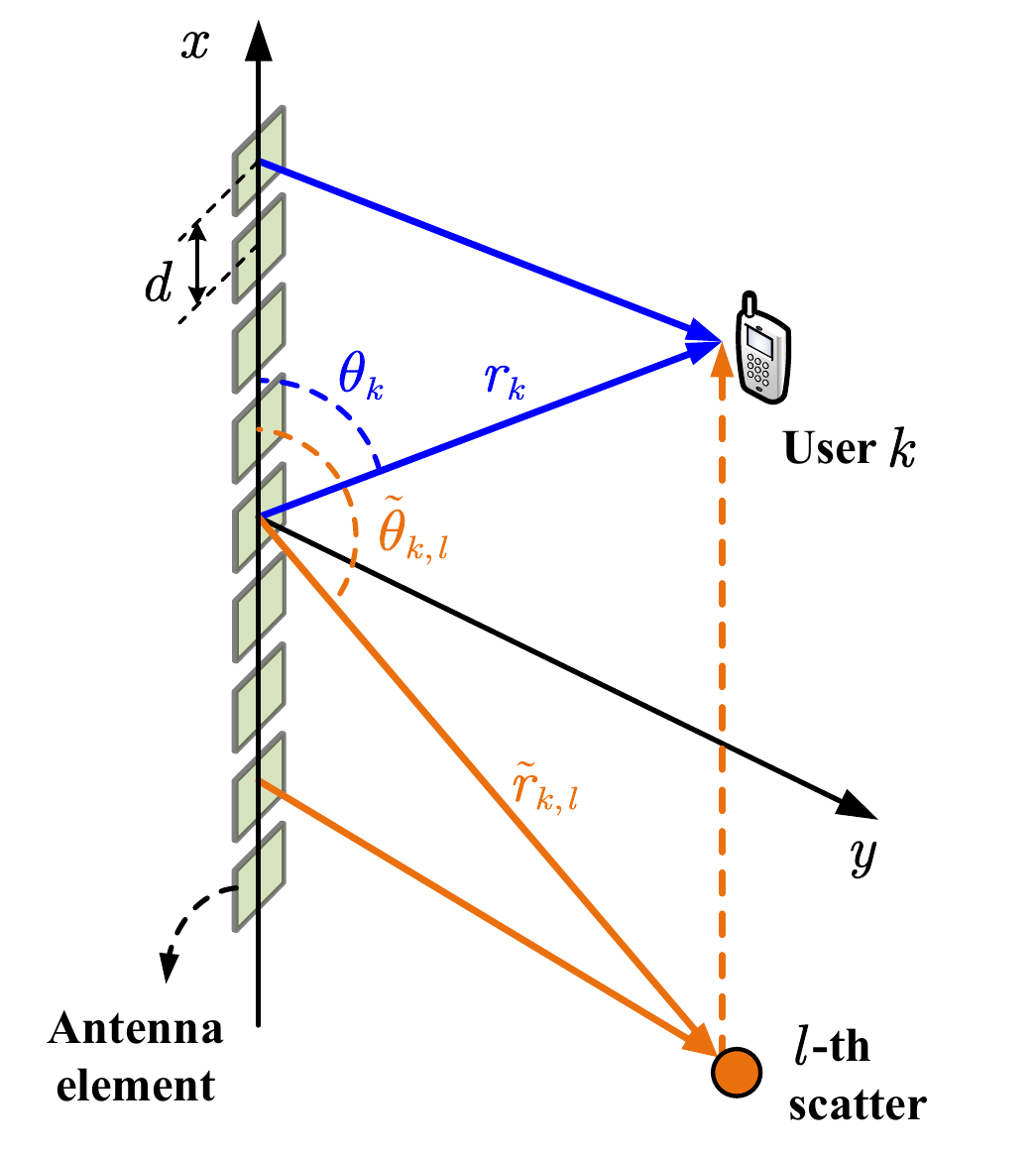}
    \caption{Illustration of the multipath near-field channel model.}
    \label{fig:multi_path}
\end{figure}

\vspace{-0.4cm}
\subsection{Near-Field Wideband Channel Model}

We assume that all users are located in the near-field region of the BS, i.e., the distances between users and the BS are less than the Rayleigh distance $2D^2/\lambda_c$ \cite{liu2023near}, where $D = (N-1)d$ denotes the aperture of the antenna array. Therefore, as shown in Fig. \ref{fig:multi_path}, we adopt the multipath near-field channel model for each communication user, which consists of a LoS channel and several non-line-of-sight (NLoS) channels induced by $L_k$ scatterers associated with user $k$.
Let $r_k$ and $\theta_k$ denote the distance and angle of user $k$ with respect to the central element of the ULA, respectively, and $\tilde{r}_{k,l}$ and $\tilde{\theta}_{k,l}$ denote the counterpart of the $l$-th scatterer for user $k$. Based on the uniform spherical wave model \cite{liu2023near}, the channel between the BS and user $k$ at the $m$-th subcarrier can be modeled as\footnote{Here, we assume that the discrete Fourier transform (DFT) operation and cyclic prefix removal have been conducted at the OFDM receiver. We refer to \cite{wang2018spatial} and \cite{wang2024performance} for the detailed derivation of the frequency-domain channels in OFDM systems.}
\begin{equation}
    \mathbf{h}_{m,k} = \beta_{m,k} \mathbf{b}^*(f_m,\theta_k, r_k) + \sum_{l=1}^{L_k} \tilde{\beta}_{m,k,l} \mathbf{b}^*(f_m,\tilde{\theta}_{k,l}, \tilde{r}_{k,l}).
\end{equation}     
Here, $\beta_{m,k}$ and $\tilde{\beta}_{m,k,l}$ denote the complex channel gain for LoS and NLoS channels, respectively. The channel gain for the LoS channel can be modeled as \cite{jornet2011channel} 
\begin{equation}
    |\beta_{m,k}| = \frac{c}{4 \pi f_m r_k} e^{-\frac{1}{2}k_{\mathrm{abs}}(f_m) r_k} \triangleq \alpha_k(f_m),
\end{equation}
where $k_{\mathrm{abs}}(f_m)$ represents the medium absorption coefficient at frequency $f_m$. For NLoS channels, the reflection coefficients have to be considered. This coefficient for the $l$-th NLoS channel for user $k$ at frequency $f_m$ is given by \cite{dovelos2021channel,piesiewicz2007scattering}  
\begin{equation}
    \Gamma_{k,l}(f_m) = \frac{ \cos \phi_{k,l}^{\mathrm{in}} - n_r(f_m) \cos \phi_{k,l}^{\mathrm{ref}} }{\cos \phi_{k,l}^{\mathrm{in}} + n_r(f_m) \cos \phi_{k,l}^{\mathrm{ref}}} e^{- \frac{1}{2} \big( \frac{4 \pi f_m \sigma_{r} \cos \phi_{k,l}^{\mathrm{in}} }{c} \big)^2},
\end{equation}
where $n_r(f_m) = Z_0/Z(f_m)$ is the frequency-dependent refractive index with $Z_0 = 377$ $\Omega$ and $Z(f_m)$ denoting the wave impedances of the free space and the reflecting material, respectively, $\phi_{k,l}^{\mathrm{in}}$ denotes the angle of incidence, $\phi_{k,l}^{\mathrm{ref}} = \arctan (n_r^{-1} \sin \phi_{k,l}^{\mathrm{in}} )$ denotes the angle of refraction, and $\sigma_r$ represents the roughness of the reflection material. Then, the channel gain for the NLoS channel is given by
\begin{equation}
    |\tilde{\beta}_{m,k,l}| = |\Gamma_{k,l}(f_m)| \alpha_k(f_m).  
\end{equation}        
We note that in THz bands, substantial reflection losses lead to NLoS channels experiencing attenuation greater than $10$ dB on average when compared to LoS channels \cite{priebe2013stochastic}. Consequently, the communication channel is LoS-dominant and NLoS-assisted.
Furthermore, vector $\mathbf{b}(f, \theta, r)$ denotes the array response vector and is given by   
\begin{equation} \label{near_response}
    \mathbf{b}(f, \theta, r) = \left[e^{-j \frac{2 \pi f}{c} (r_1(r, \theta)-r) },\dots,e^{-j \frac{2 \pi f}{c} (r_N(r, \theta)-r)} \right]^T,
\end{equation}
where $r_n(r, \theta)$ denote the propagation distance of the signal emitted from the $n$-th antenna. In the near-field region, it has to be accurately modeled as 
\begin{equation}
    r_n(r, \theta) = \sqrt{ r^2 + \chi_n^2d^2 - 2 r \chi_n d\cos \theta },
\end{equation}  
where $\chi_n = n -1 - \frac{N-1}{2}$. It is worth noting that the near-field array response vector is frequency-dependent due to the spatial wideband effect. This property introduces additional challenges for near-field wideband beamforming design.

\vspace{-0.3cm}
\subsection{Antenna Architectures for Wideband Beamfocusing}
As highlighted in Section \ref{sec:intro}, TTD-based hybrid beamforming architectures present a viable solution to address challenges arising from the frequency-dependent array response. Unlike conventional PSs\footnote{Due to the physical properties of the components used in implementing PSs (like capacitors and inductors), which have frequency-dependent behaviors, the response of PSs naturally changes with frequencies. However, unlike TTDs, the frequency-dependent behavior of PSs is fixed according to the circuit design and cannot be dynamically controlled. Such a property may complicate the system design. Therefore, in wideband systems, the PSs are usually designed to provide an approximately flat response across the entire working frequency band \cite{5648370, 4623722}. Therefore, in this paper, we assume that the response of PSs is perfectly frequency-flat as in the literature \cite{ gao2021wideband, dai2022delay}.}, TTDs can realize a frequency-dependent phase shift of $e^{-j 2\pi f_m t}$ on subcarrier $m$ in the frequency domain by introducing a time delay $t$ into the signals \cite{6531062}. However, due to the high hardware cost and energy consumption of TTDs, TTD-based hybrid beamforming architectures usually install only a limited number of TTDs between the RF chain and the PS network \cite{gao2021wideband,dai2022delay}. In the conventional fully-connected architecture, as illustrated in Fig. \ref{fig:TTD_full}, each TTD is linked to a sub-array through PSs and each RF chain is connected to the entire antenna array. To further improve energy efficiency and reduce TTD hardware requirements, we introduce a new sub-connected setup for TTD-based hybrid beamforming, depicted in Fig. \ref{fig:TTD_sub}, where each RF chain is connected to a sub-array via TTDs and PSs. 

\begin{figure}[!t]
    \centering
    \subfigure[Fully-connected architecture.]{
        \includegraphics[width=0.4\textwidth]{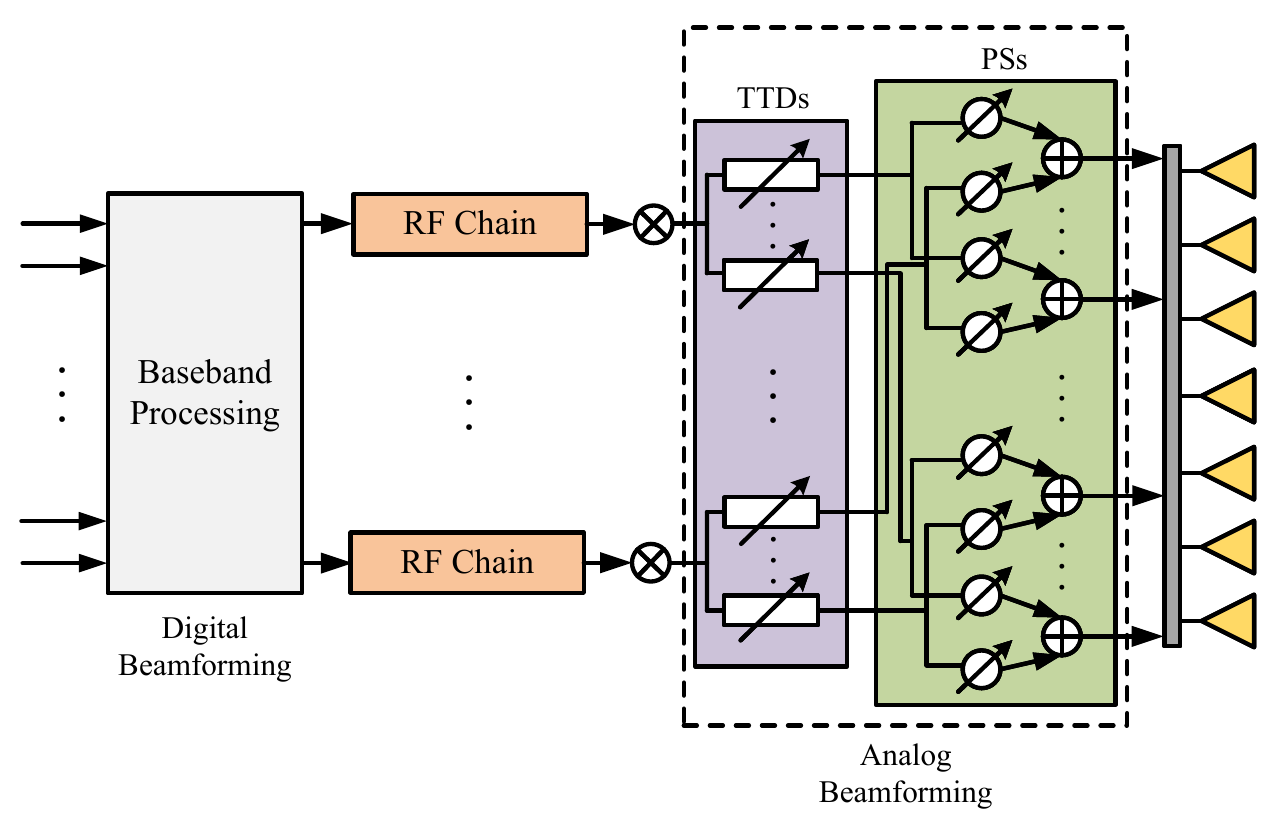}
        \label{fig:TTD_full}
    }
    \subfigure[Sub-connected architecture.]{
        \includegraphics[width=0.4\textwidth]{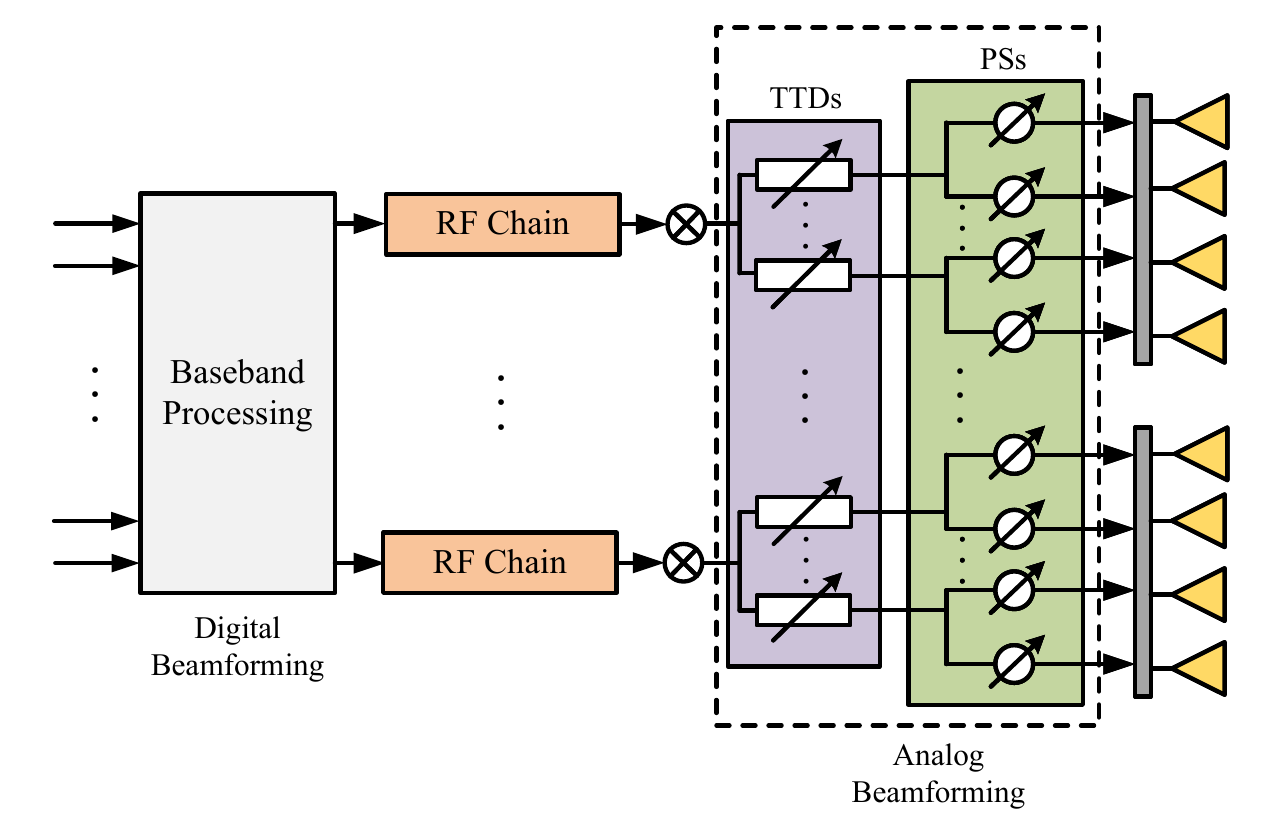}
        \label{fig:TTD_sub}
    }
    \caption{Architectures for TTD-based hybrid beamforming.}
    \label{fig:TTD}
\end{figure}

Let $N_{\mathrm{RF}}$ denote the number of RF chains in the TTD-based hybrid beamforming architecture, which is subject to the constraint $K \le N_{\mathrm{RF}} \ll N$, and $N_{\mathrm{T}}$ denote the number of TTDs connected to each RF chain. We assume that each TTD is connected to a sub-array with the same size. Then, when the TTD-hybrid beamforming architecture is exploited, the received signal of user $k$ on subcarrier $m$ is given by
\begin{equation} \label{eqn:received_signal}
    y_{m,k} = \mathbf{h}_{m,k}^H \mathbf{A} \mathbf{T}_m \mathbf{D}_m \mathbf{s}_m + n_{m,k},
\end{equation}  
where matrices $\mathbf{A} \in \mathbb{C}^{N \times N_{\mathrm{T}}N_{\mathrm{RF}}}$ and $\mathbf{T}_m \in \mathbb{C}^{N_{\mathrm{T}}N_{\mathrm{RF}} \times N_{\mathrm{RF}}}$ denote the frequency-independent and frequency-dependent analog beamformers realized by TTDs and PSs, respectively, matrix $\mathbf{D}_m \in \mathbb{C}^{N_{\mathrm{RF}} \times K}$ denotes the baseband digital beamformer, vector $\mathbf{s}_m \in \mathbb{C}^{K \times 1}$ denotes the information symbols of the $K$ users on subcarrier $m$, satisfying $\mathbb{E}[ \mathbf{s}_m \mathbf{s}_m^H ] = \mathbf{I}_K$, and $n_{m,k} \sim \mathcal{CN}(0, \sigma_{m,k}^2)$ represents the additive Gaussian white noise. For both fully-connected and sub-connected architectures, the TTD-based analog beamformer is given by 
\begin{equation} \label{TTD_analog}
    \mathbf{T}_m = \mathrm{blkdiag} \left( e^{-j 2\pi f_m  \mathbf{t}_1},\dots,e^{-j 2\pi f_m  \mathbf{t}_{N_{\mathrm{RF}}}} \right),
\end{equation}   
where $\mathbf{t}_n = [t_{n,1},\dots,t_{n, N_{\mathrm{T}}}]^T \in \mathbb{R}^{N_{\mathrm{T}} \times 1}$ denotes the time-delay vector realized the TTDs connected to the $n$-th RF chain. The time delay of each TTD needs to satisfy the maximum delay constraint, i.e., $t_{n,l} \in [0, t_{\max}], \forall n = 1,...,N_{\mathrm{RF}}, l = 1,\dots,N_{\mathrm{T}}$. Typically, achieving a larger $t_{\max}$ requires a more complex structure and a larger footprint size of TTDs \cite{yan2022energy, 8248806}.
The PS-based analog beamformer $\mathbf{A}$ has different structures for the fully-connected and sub-connected architectures, which are detailed in the following.
\begin{itemize}
    \item \textbf{Fully-connected Architecture:} In this architecture, each TTD is connected to $N/N_{\mathrm{T}}$ antennas via PSs. Therefore, the PS-based analog beamformer $\mathbf{A}$ for the fully-connected architecture is given by
    \begin{align}
        \label{ps_full}
        &\mathbf{A} = [\mathbf{A}_1 ,\dots,\mathbf{A}_{N_{\mathrm{RF}}} ], \\
        \label{ps_full_2}
        &\mathbf{A}_n  = \mathrm{blkdiag} ( \mathbf{a}_{n,1} ,...,\mathbf{a}_{n,N_{\mathrm{T}}}  ),
    \end{align} 
    where $\mathbf{A}_{n}  \in \mathbb{C}^{N \times N_{\mathrm{T}}}$ and $\mathbf{a}_{n,l} \in \mathbb{C}^{\frac{N}{N_{\mathrm{T}}} \times 1}$ denote the PS-based analog beamformer connected to the $n$-th RF chain and the $l$-th TTD of this chain, respectively.
    
    \item \textbf{Sub-connected Architecture:} In this architecture, each RF chain is only connected to a sub-array of $N_{\mathrm{sub}} = N/N_{\mathrm{RF}}$ antennas, and thus each TTD is connected to $N_{\mathrm{sub}}/N_{\mathrm{T}}$ antennas. As a result, the corresponding PS-based analog beamformer $\mathbf{A} $ can be expressed follows:
    \begin{align}
        \label{ps_sub}
        &\mathbf{A} = \mathrm{blkdiag}(\mathbf{A}_{1},\dots,\mathbf{A}_{N_{\mathrm{RF}}}), \\
        \label{ps_sub_2}
        &\mathbf{A}_{n} = \mathrm{blkdiag}( \mathbf{a}_{n,1},\dots,\mathbf{a}_{n, N_{\mathrm{T}}} ),
    \end{align}
    where $\mathbf{A}_{n} \in \mathbb{C}^{N_{\mathrm{sub}} \times N_{\mathrm{T}}}$ and $\mathbf{a}_{n,l} \in \mathbb{C}^{\frac{N_{\mathrm{sub}}}{N_{\mathrm{T}}}\times 1}$ denote the PS-based analog beamformer for the sub-array connected to the $n$-th RF chain and its $l$-th TTD, respectively.
\end{itemize}
Moreover, due to the inherent limitations of PSs, which can only adjust the phase of the signals, the analog beamformer based on PSs must satisfy the unit-modulus constraint. In other words, for non-zero entries, the PS-based analog beamformer should satisfy $|\mathbf{A}(i,j)| = 1$.

\begin{remark} \label{remark_1}
    \emph{
    Compared to fully-connected architectures, sub-connected architectures have the following advantages. \emph{First}, it significantly reduces the required number of PSs from $N N_{\mathrm{RF}}$ to $N$. \emph{Second}, in the sub-connected architecture, the TTDs connected to each RF chain only need to mitigate the spatial wideband effect of individual sub-arrays, which is less significant compared to the effect related to the entire array. Therefore, the requirements on the number of TTDs and the maximum delay are more relaxed.}
\end{remark}

\vspace{-0.4cm}
\subsection{Problem Formulation}
We aim to maximize the spectral efficiency of the considered near-field wideband communication system by jointly optimizing the analog and digital beamformers of the TTD-based hybrid beamforming architecture. According to \eqref{eqn:received_signal}, the signal-to-interference-plus-noise ratio (SINR) for decoding the desired signal of user $k$ on subcarrier $m$ is given by 
\begin{equation} \label{achievable_rate}
    \gamma_{m,k} = \frac{| \mathbf{h}_{m,k}^H \mathbf{A} \mathbf{T}_m \mathbf{d}_{m,k} |^2}{\sum_{i =1, i \neq k}^K | \mathbf{h}_{m,k}^H \mathbf{A} \mathbf{T}_m \mathbf{d}_{m,i} |^2 + \sigma_{m,k}^2},
\end{equation}    
where $\mathbf{d}_{m,k} \in \mathbb{C}^{N_{\mathrm{RF}} \times 1}$ is the vector at the $k$-th column of $\mathbf{D}_m$ and also the digital beamformer of user $k$ on subcarrier $m$. Therefore, by taking into account the OFDM cyclic prefix, the spectral efficiency of the considered system is given by
\begin{equation}
    R = \frac{1}{M + L_{\mathrm{CP}}} \sum_{m=1}^M \sum_{k=1}^K \log_2 (1 + \gamma_{m,k}).
\end{equation}
Therefore, the spectral efficiency maximization problem can be formulated as follows:
\begin{subequations} \label{problem_se_max}
    \begin{align}
        \max_{\mathbf{A}, \mathbf{T}_m, \mathbf{D}_m}  \quad & \sum_{m=1}^M \sum_{k=1}^K \log_2 (1 + \gamma_{m,k}) \\
        \label{constraint:1}
        \mathrm{s.t.} \quad & \|\mathbf{A} \mathbf{T}_m \mathbf{D}_m \|_F^2 \le P_t, \forall m,  \\
        \label{constraint:2}
        & |\mathbf{A}(i,j)| = 1, \forall (i,j) \in \mathcal{A}, \\
        \label{constraint:3}
        & \mathbf{T}_m = \mathrm{blkdiag} \big( e^{-j 2\pi f_m  \mathbf{t}_1}, \nonumber \\
        &\hspace{1.6cm} \dots,e^{-j 2\pi f_m  \mathbf{t}_{N_{\mathrm{RF}}}} \big), \forall m,\\
        \label{constraint:4}
        & t_{n,l} \in [0, t_{\max}], \forall n, l,
    \end{align}
\end{subequations}
where $P_t$ denotes the maximum transmit power available for each subcarrier \cite{sohrabi2016hybrid} and set $\mathcal{A}$ collects indexes of non-zero entry in matrix $\mathbf{A}$. Problem \eqref{problem_se_max} is generally challenging to solve due to the high coupling between different optimization variables and different subcarriers. To the best of the authors' knowledge, there is no generic solution to the problem \eqref{problem_se_max} in the presence of the above challenges. Therefore, in the following, we first propose an FDA approach to this optimization problem, which has provable optimality but incurs higher computational complexity and higher channel estimation overhead. This approach provides theoretical performance bounds for TTD-based beamforming architectures. Then, an HTS approach is proposed to reduce the complexity.

\section{Fully-digital Approximation Approach} \label{sec:FDA}

In this section, we develop an FDA approach, which directly maximizes the spectral efficiency, to solve problem \eqref{problem_se_max}. The main idea of the FDA approach is to optimize the hybrid beamformer sufficiently “close” to the optimal unconstrained fully-digital beamformer through the penalty method \cite{shi2020penalty}. 

\subsection{Problem Reformulation}
In this subsection, we reformulate problem \eqref{problem_se_max} into an FDA problem. To streamline the optimization process, we first remove the transmit power constraint in problem \eqref{problem_se_max}. Based on the fact that the optimal multi-user beamformers tend to use all available transmit power to maximize the special efficiency \cite{6832894}, we have the following lemma.

\begin{lemma} \label{lemma_1}
    \emph{
        Let $\tilde{\mathbf{A}}$, $\tilde{\mathbf{T}}_m$, and $\tilde{\mathbf{D}}_m$ denote an optimal solution to the following optimization problem:  
        \begin{subequations} \label{problem_se_max_new}
            \begin{align}
                \max_{\mathbf{A}, \mathbf{T}_m, \mathbf{D}_m}  \quad & \sum_{m=1}^M \sum_{k=1}^K \log_2 (1 + \tilde{\gamma}_{m,k}) \\
                \mathrm{s.t.} \quad & \eqref{constraint:2}-\eqref{constraint:4},
            \end{align}
        \end{subequations} 
        where $\tilde{\gamma}_{m,k}$ is given by 
        \begin{equation}
            \tilde{\gamma}_{m,k} = \frac{| \mathbf{h}_{m,k}^H \mathbf{A} \mathbf{T}_m \mathbf{d}_{m,k} |^2}{\sum_{i =1, i \neq k}^K | \mathbf{h}_{m,k}^H \mathbf{A} \mathbf{T}_m \mathbf{d}_{m,i} |^2 + \frac{\sigma_{m,k}^2}{P_t} \|\mathbf{A} \mathbf{T}_m \mathbf{D}_m \|_F^2}.
        \end{equation} 
        Then, an optimal solution to problem \eqref{problem_se_max} can be obtained by 
        \begin{equation} \label{solution_scale}
            \mathbf{A} = \tilde{\mathbf{A}}, \mathbf{T}_m = \tilde{\mathbf{T}}_m, \mathbf{D}_m = \frac{\sqrt{P_t}}{\|\tilde{\mathbf{A}} \tilde{\mathbf{T}}_m \tilde{\mathbf{D}}_m \|_F} \tilde{\mathbf{D}}_m.
        \end{equation}}
\end{lemma}

\begin{proof}
    It is easy to verify that the solution in \eqref{solution_scale} satisfies all constraints in problem \eqref{problem_se_max} and achieves the same objective value of problem \eqref{problem_se_max_new} as the original solution $\tilde{\mathbf{A}}$, $\tilde{\mathbf{T}}_m$, and $\tilde{\mathbf{D}}_m$. Therefore, the solution in \eqref{solution_scale} guarantees the maximization of the objective function of problem \eqref{problem_se_max} under the condition that the power constraint is satisfied with equality.
\end{proof}

Although we can use \textbf{Lemma \ref{lemma_1}} to remove the transmit power constraint, the other constraints on matrices $\mathbf{A}$ and $\mathbf{T}_m$ are still challenging to tackle since these matrices are still highly coupled in the sum-of-logarithms objective function. As a remedy, we introduce unconstrained fully-digital beamformers $\mathbf{W}_m = [\mathbf{w}_{m,1},\dots,\mathbf{w}_{m,K}] \in \mathbb{C}^{N \times K}, \forall m,$ as auxiliary variables. Ideally, the equality $\mathbf{W}_m = \mathbf{A} \mathbf{T}_m \mathbf{D}_m$ is expected, leading to the following equivalent form of problem \eqref{problem_se_max_new}:
\begin{subequations} \label{problem_se_max_new_2}
    \begin{align}
        \max_{\mathbf{A}, \mathbf{T}_m, \mathbf{D}_m, \mathbf{W}_m}  & \sum_{m=1}^M \sum_{k=1}^K \log_2 (1 + \hat{\gamma}_{m,k}(\mathbf{W}_m)) \\
        \mathrm{s.t.} \quad & \mathbf{W}_m = \mathbf{A} \mathbf{T}_m \mathbf{D}_m, \forall m, \\
        & \eqref{constraint:2}-\eqref{constraint:4},
    \end{align}
\end{subequations} 
where $\hat{\gamma}_{m,k}(\mathbf{W}_m)$ is the modified SINR realized by the fully-digital beamformer and is given by
\begin{equation}
    \hat{\gamma}_{m,k}(\mathbf{W}_m) = \frac{| \mathbf{h}_{m,k}^H \mathbf{w}_{m,k} |^2}{\sum_{i =1, i \neq k}^K | \mathbf{h}_{m,k}^H \mathbf{w}_{m,i} |^2 + \frac{\sigma_{m,k}^2}{P_t} \|\mathbf{W}_m \|_F^2}.
\end{equation} 
However, the new equality constraint in problem \eqref{problem_se_max_new_2} still hinders the optimization of the hybrid beamformers. To address this issue, the penalty method can be employed, leading to the following optimization problem 
\begin{subequations} \label{problem_se_max_new_3}
    \begin{align}
        \max_{\mathbf{A}, \mathbf{T}_m, \mathbf{D}_m, \mathbf{W}_m}  & \sum_{m=1}^M \sum_{k=1}^K \log_2 (1 + \hat{\gamma}_{m,k}(\mathbf{W}_m)) \nonumber \\
        &- \frac{1}{\rho} \sum_{m=1}^M \|\mathbf{W}_m - \mathbf{A} \mathbf{T}_m \mathbf{D}_m\|_F^2, \\
        \mathrm{s.t.} \quad & \eqref{constraint:2}-\eqref{constraint:4},
    \end{align}
\end{subequations} 
where $\rho > 0$ is the penalty factor. In this problem, the hybrid beamformer can be optimized to gradually approximate the fully-digital beamformer by gradually reducing the value of $\rho$. In particular, when $\rho \rightarrow 0$, the penalty term can be exactly zero, which yields a solution to problem \eqref{problem_se_max_new_2}. However, to ensure optimal performance, the penalty factor is typically initialized with a large value, ensuring that the original objective function is sufficiently maximized. Subsequently, the penalty factor is gradually decreased to guarantee the penalty term approaches zero, i.e., the equality constraint is achieved. Based on the results in \cite{shi2020penalty}, it can be proved that the penalty method guarantees the convergence to a stationary point of problem \eqref{problem_se_max_new_2}, which is also a stationary point of problem \eqref{problem_se_max_new}. In the following, we will provide solutions to problem \eqref{problem_se_max_new_3} for the fully-connected and sub-connected architectures, respectively.

\subsection{Solution for Fully-connected Architectures} \label{sec: FDA_full}
In problem \eqref{problem_se_max_new_3}, all constraints are related to the optimization variables $\mathbf{A}$ and $\mathbf{T}_m, \forall m$. This observation motivates us to apply the block coordinate descent (BCD) method, which alternately optimizes each block with the other blocks fixed in each iteration, to solve this problem. In the following, we elaborate on the optimization of each block.

\subsubsection{Subproblem With Respect to $\mathbf{A}$ and $\mathbf{T}_m$} 

The variables $\mathbf{A}$ and $\mathbf{T}_m$ only appear in the last term of the objective function, leading to the following optimization problem:
\begin{subequations}
    \begin{align}
        \label{obj_1}
        \min_{\mathbf{A}, \mathbf{T}_m}  \quad & \sum_{m=1}^M \left\|\mathbf{W}_m-  \mathbf{A} \mathbf{T}_m \mathbf{D}_m \right\|_F^2 \\
        \mathrm{s.t.} \quad & \eqref{constraint:2}-\eqref{constraint:4}.
    \end{align}
\end{subequations}
For fully-connected architectures, the analog beamformer matrix $\mathbf{A} \mathbf{T}_m$ is a full matrix, resulting in the coupling between the time delay of different TTDs. To address this problem, we introduce additional auxiliary variables $\mathbf{V}_m, \forall m,$ that satisfies $\mathbf{V}_m = \mathbf{A} \mathbf{T}_m$. In the following, we will show that the introduction of $\mathbf{V}_m$ can facilitate the closed-form update of the analog beamformer coefficients. Similarly, this equality constraint can be moved to the objective function as a penalty term, resulting in the following optimization problem:
\begin{subequations} \label{problem_A_T_V}
    \begin{align}
        \label{obj_1_penalty}
        \min_{\mathbf{A}, \mathbf{T}_m, \mathbf{V}_m}  & \sum_{m=1}^M \left( \left\|\mathbf{W}_m-  \mathbf{V}_m \mathbf{D}_m \right\|_F^2 + \frac{1}{\bar{\rho}}  \|\mathbf{V}_m - \mathbf{A} \mathbf{T}_m\|_F^2 \right),  \\
        \mathrm{s.t.} \quad & \eqref{constraint:2}-\eqref{constraint:4},
    \end{align}
\end{subequations}
where $\bar{\rho} > 0$ is the penalty factor for the equality constraint $\mathbf{V}_m = \mathbf{A} \mathbf{T}_m$. Similarly, when $\bar{\rho} \rightarrow 0$, this equality constraint is guaranteed. To decouple the optimization variables, we apply BCD to solve this new optimization problem to alternately update each block as presented in the following.

By fixing $\mathbf{T}_m$ and $\mathbf{V}_m$, problem \eqref{problem_A_T_V} reduces to 
\begin{subequations}
    \begin{align}
        \label{obj_A_1}
        \min_{\mathbf{A}}  \quad & \sum_{m=1}^M \left\|\mathbf{V}_m-  \mathbf{A} \mathbf{T}_m \right\|_F^2 \\
        \mathrm{s.t.} \quad & |\mathbf{A}(i,j)| = 1, \forall (i,j) \in \mathcal{A}.
    \end{align}
\end{subequations}  
Based on \eqref{ps_full} and \eqref{ps_full_2}, the objective function can be equivalently written as 
\begin{align} \label{simple_obj}
    \eqref{obj_A_1} = &\sum_{n = 1}^{N_\mathrm{RF}} \sum_{l=1}^{N_{\mathrm{T}}} \sum_{m=1}^M \left\|\mathbf{v}_{m,n,l} -  \mathbf{a}_{n,l} e^{-j 2\pi f_m t_{n,l}} \right\|^2 \nonumber \\
    = &\sum_{n = 1}^{N_\mathrm{RF}} \sum_{l=1}^{N_{\mathrm{T}}} \sum_{m=1}^M \Big(\mathbf{a}_{n,l}^H \mathbf{a}_{n,l} + \mathbf{v}_{m,n,l}^H \mathbf{v}_{m,n,l} \nonumber \\
    & \hspace{1.3cm}- 2 \Re \left\{ \mathbf{v}_{m,n,l}^H \mathbf{a}_{n,l} e^{-j 2 \pi f_m t_{n,l}} \right\}  \Big),
\end{align}
where $\mathbf{v}_{m,n,l} = \mathbf{V}_m \big((l-1)\frac{N}{N_{\mathrm{T}}}+1:l\frac{N}{N_{\mathrm{T}}}, n \big)$. There are two key observations for the above new objective function. Firstly, according to the unit-modulus constraint on $\mathbf{a}_{n,l}$, we have $\mathbf{a}_{n,l}^H \mathbf{a}_{n,l} = MN/N_{\mathrm{T}}$, which is a constant. Therefore, the optimization variable is only related to the last term. Secondly, vectors $\mathbf{a}_{n,l}, \forall n, l,$ are no longer coupled with each other and can be optimized independently. Based on these observations, for each vector $\mathbf{a}_{n,l}$, its optimal solution can be obtained by solving the following optimization problem:
\begin{subequations}
    \begin{align}
        \max_{\mathbf{a}_{n,l}} \quad & \sum_{m=1}^M \Re \left\{ \mathbf{v}_{m,n,l}^H \mathbf{a}_{n,l} e^{-j 2 \pi f_m t_{n,l}} \right\} \\
        \mathrm{s.t.} \quad &  | \mathbf{a}_{n,l}(i) | = 1, \forall i.
    \end{align}
\end{subequations}
It can be readily obtained that the optimal solution is
\begin{equation} \label{opt_solution_PS}
    \mathbf{a}_{n,l} = e^{j \angle \left(\sum_{m=1}^M \mathbf{v}_{m,n,l} e^{j 2 \pi f_m t_{n,l}}\right)}.
\end{equation}

Next, we optimize the TTD-based analog beamformer $\mathbf{T}_m, \forall m,$ by fixing the other optimization variables. The resulting optimization is given by
\begin{subequations} \label{problem_TTD}
    \begin{align}
        \min_{\mathbf{T}_m }  \quad & \sum_{m=1}^M \|\mathbf{V}_m -  \mathbf{A} \mathbf{T}_m \|_F^2 \\
        \mathrm{s.t.} \quad & \eqref{constraint:3}, \eqref{constraint:4}.
    \end{align}
\end{subequations}
According to the result in \eqref{simple_obj}, the above optimization problem is also a separable problem with respect to each time delay $t_{n,l}, \forall n, l$, leading to the following independent optimization problem for each $t_{n,l}$: 
\begin{subequations} \label{problem_TTD_3}
    \begin{align}
        \max_{t_{n,l} }  \quad & \sum_{m=1}^M \Re \left\{ \mathbf{v}_{m,n,l}^H \mathbf{a}_{n,l} e^{-j 2 \pi f_m t_{n,l}} \right\}\\
        \mathrm{s.t.} \quad & t_{n,l} \in [0, t_{\max}].
    \end{align}
\end{subequations}
This problem is a classical single-variable optimization problem within a fixed interval, which can be effectively solved using one-dimensional search. Specifically, by defining a searching set $\mathcal{Q} = \{0, t_{\max}/(Q-1), 2t_{\max}/(Q-1),\dots,t_{\max} \}$ with $Q$ samples, a nearly-optimal $t_{n,l}$ can be obtained as 
\begin{equation} \label{opt_solution_TTD}
    t_{n,l} = \arg \max_{t_{n,l} \in \mathcal{Q}} \Re \left\{ \mathbf{v}_{m,n,l}^H \mathbf{a}_{n,l} e^{-j 2 \pi f_m t_{n,l}} \right\}.
\end{equation}   

Furthermore, given $\mathbf{A}$ and $\mathbf{T}_m$, the optimal solution of $\mathbf{V}_m$ can be obtained by minimizing objective function of problem \eqref{problem_A_T_V} without constraint, which is given by 
\begin{align} \label{opt_solution_V}
    \mathbf{V}_m = &\arg \min_{\mathbf{V}_m} \| \mathbf{W}_m - \mathbf{V}_m \mathbf{D}_m \|_F^2 + \frac{1}{\bar{\rho}} \|\mathbf{V}_m - \mathbf{A} \mathbf{T}_m \|_F^2 \nonumber \\
    = & \left(\mathbf{W}_m \mathbf{D}_m^H + \frac{1}{\bar{\rho}} \mathbf{A} \mathbf{T}_m \right) \left( \mathbf{D}_m \mathbf{D}_m^H + \frac{1}{\bar{\rho}} \mathbf{I}_{N_\mathrm{RF}} \right)^{-1}.
\end{align} 

Given the above solution to each block, the proposed penalty-based method for solving \eqref{problem_A_T_V} is summarized in \textbf{Algorithm \ref{alg:penalty}}. The convergence of this algorithm is guaranteed since each block is updated by the optimal solution in each iteration and the objective function is bounded from below. 

\begin{algorithm}[tb]
    \caption{Penalty-based algorithm for solving \eqref{problem_A_T_V}.}
    \label{alg:penalty}
    \begin{algorithmic}[1]
        \STATE{initialize $\mathbf{A}$, $\mathbf{T}_m$, $\mathbf{V}_m$, $0<c<1$, and $\bar{\rho} > 0$}
        \REPEAT
        \REPEAT
        \STATE{update $\mathbf{A}$ by \eqref{opt_solution_PS}}
        \STATE{update $\mathbf{T}_m, \forall m,$ by \eqref{opt_solution_TTD}}
        \STATE{update $\mathbf{V}_m, \forall m,$ by \eqref{opt_solution_V}}
        \UNTIL{the fractional decrease of the objective value of problem \eqref{problem_A_T_V} falls below a threshold}
        \STATE{update penalty factor as $\bar{\rho} = c \bar{\rho}$ }
        \UNTIL{the penalty value falls below a threshold}
    \end{algorithmic}
\end{algorithm}

\subsubsection{Subproblem With Respect to $\mathbf{D}_m$}
The optimal solution of $\mathbf{D}_m$ can be obtained by minimizing the second term of the objective function in problem \eqref{problem_se_max_new_3} without constraint, which is given by 
\begin{align} \label{opt_solution_D}
    \mathbf{D}_m = &\arg \min_{\mathbf{D}_m} \|\mathbf{W}_m - \mathbf{A}_m \mathbf{T}_m \mathbf{D}_m\|_F^2 \nonumber \\
    = &(\mathbf{T}_m^H \mathbf{A}^H \mathbf{A} \mathbf{T}_m)^{-1} \mathbf{T}_m^H \mathbf{A}^H \mathbf{W}_m.
\end{align}

\subsubsection{Subproblem With Respect to $\mathbf{W}_m$} \label{sec:W_opt}
When other blocks are fixed, the optimization problem for $\mathbf{W}_m$ is given by 
\begin{align}  \label{problem_P}
    \max_{\mathbf{W}_m} \quad & \sum_{m=1}^M \sum_{k=1}^K \log_2 (1 + \hat{\gamma}_{m,k}(\mathbf{W}_m)) \nonumber \\
    &- \frac{1}{\rho} \sum_{m=1}^M \|\mathbf{W}_m - \mathbf{A} \mathbf{T}_m \mathbf{D}_m\|_F^2.
\end{align}
While this problem also has no constraint, its non-convex objective function makes it still challenging to solve. To address this issue, we adopt the quadratic transform \cite[Theorem 2]{shen2018fractional} and Lagrangian dual transform \cite[Theorem 3]{shen2018fractional2} to convert problem \eqref{problem_P} into the following equivalent form:
\begin{align}  \label{problem_P_2}
    &\max_{\mathbf{W}_m, \boldsymbol{\lambda}_m, \boldsymbol{\mu}_m} \sum_{m=1}^M \sum_{k=1}^K \Bigg(2 \sqrt{1 + \mu_{m,k}} \Re  \left\{ \lambda_{m,k}^* \mathbf{h}_{m,k}^H \mathbf{w}_{m,k} \right\} \nonumber \\
    & - |\lambda_{m,k}|^2 \left( \sum_{i=1}^K | \mathbf{h}_{m,k}^H \mathbf{w}_{m,i} |^2 + \frac{\sigma_{m,k}^2}{P_t} \|\mathbf{W}_m \|_F^2 \right)  \Bigg) \nonumber \\
    & - \frac{1}{\rho} \sum_{m=1}^M \|\mathbf{W}_m - \mathbf{A} \mathbf{T}_m \mathbf{D}_m\|_F^2 \triangleq g(\mathbf{W}_m, \boldsymbol{\lambda}_m, \boldsymbol{\mu}_m),
\end{align}
where $\boldsymbol{\lambda}_m = [\lambda_{m,1}, \lambda_{m,2}, \dots, \lambda_{m,K}]^T$ and $\boldsymbol{\mu}_m = [\mu_{m,1}, \mu_{m,2}, \dots, \mu_{m,K}]^T$ are auxiliary variables. Following \cite{shen2018fractional} and \cite{shen2018fractional2}, for any given $\mathbf{W}_m$, the values of $\mu_{m,k}$ and $\lambda_{m,k}$ can be sequentially updated by the following closed-form solution:
\begin{align}
    \label{opt_solution_mu}
    \mu_{m,k} = &\hat{\gamma}_{m,k}(\mathbf{W}_m), \forall m,k, \\
    \label{opt_solution_lambda}
    \lambda_{m,k} = &\frac{\sqrt{1 + \mu_{m,k}} \mathbf{h}_{m,k}^H \mathbf{w}_{m,k} }{ \sum_{i =1}^K | \mathbf{h}_{m,k}^H \mathbf{w}_{m,i} |^2 + \frac{\sigma_{m,k}^2}{P_t} \|\mathbf{W}_m \|_F^2 }, \forall m,k.
\end{align}

Then, for any given $\boldsymbol{\lambda}_m$ and $\boldsymbol{\mu}_m$, problem \eqref{problem_P_2} is an unconstrained convex optimization problem for $\mathbf{W}_m$. Based on the first-order optimality condition $\frac{\partial g(\mathbf{W}_m, \boldsymbol{\lambda}_m, \boldsymbol{\mu}_m)}{\partial \mathbf{W}_m} = 0$, the optimal $\mathbf{W}_m$ is given by 
\begin{align} \label{opt_solution_W}
    &\mathbf{W}_m = \Bigg( \frac{1}{\rho} \mathbf{I}_N + \sum_{k=1}^K |\lambda_{m,k}|^2 \Big( \mathbf{h}_{m,k} \mathbf{h}_{m,k}^H + \frac{\sigma_{m,k}^2}{P_t} \mathbf{I}_N \Big) \Bigg)^{-1} \nonumber \\
    &\times \left(\frac{1}{\rho} \mathbf{A} \mathbf{T}_m \mathbf{D}_m + \sum_{k=1}^K \sqrt{1 + \mu_{m,k}} \lambda_{m,k}^* \mathbf{h}_{m,k} \mathbf{e}_k^T\right),
\end{align}    
where $\mathbf{e}_k \in \mathbb{R}^{K \times 1}$ is a vector whose $k$-th entry is one while the others are zero.

\subsubsection{Overall Algorithm, Convergence, and Complexity}
Building upon the solutions developed for each optimization block, the principles of the penalty method, and the results in \textbf{Lemma \ref{lemma_1}}, the overall penalty-based FDA algorithm for solving problem \eqref{problem_se_max} for fully-connected architectures is summarized in \textbf{Algorithm \ref{alg:penalty_FDA}}. Since the objective value is non-increasing in each step of block coordinate descent in the inner loop and is bounded from below, the convergence of \textbf{Algorithm \ref{alg:penalty_FDA}} is guaranteed. Furthermore, \textbf{Algorithm \ref{alg:penalty_FDA}} is computationally efficient, as the optimization variables in each step are updated by either the closed-form solution or the low-complexity one-dimensional search. It can be shown that the complexities of updating $\mathbf{A}$, $\mathbf{T}_m$, and $\mathbf{V}_m$ in \textbf{Algorithm \ref{alg:penalty}} are in order of $\mathcal{O}(M N)$, $\mathcal{O}(M N^2 N_{\mathrm{RF}} Q / N_{\mathrm{T}})$, and $\mathcal{O}(M(NKN_{\mathrm{RF}} + (NN_{\mathrm{T}} +K ) N_{\mathrm{RF}}^2 + N_{\mathrm{RF}}^3))$, respectively, where $\mathcal{O}(\cdot)$ is the big-O notation. Furthermore, the complexities for updating $\mathbf{D}_m$, $\boldsymbol{\mu}_{m,k}$, $\boldsymbol{\lambda}_{m,k}$, and $\mathbf{W}_m$ in \textbf{Algorithm \ref{alg:penalty_FDA}} are in order of $\mathcal{O}(M (2N N_{\mathrm{RF}}^2 +  N_{\mathrm{RF}}^3 + N N_{\mathrm{RF}} K ))$, $\mathcal{O}(MK(N + K(K-1)N + N^2 ))$, $\mathcal{O}(MK((K^2+1)N + N^2 ))$, and $\mathcal{O}( M( (KN_{\mathrm{RF}} + K^2) N + K N^2 + N^3) )$, respectively.

\begin{algorithm}[tb]
    \caption{Penalty-FDA Approach for Fully-connected Architectures.}
    \label{alg:penalty_FDA}
    \begin{algorithmic}[1]
        \STATE{initialize $\mathbf{A}$, $\mathbf{T}_m$, $\mathbf{D}_m$, $\mathbf{V}_m$  $\mathbf{W}_m$, $0<c<1$, and $\rho > 0$}
        \REPEAT
        \REPEAT
        \STATE{update $\mathbf{A}$ and $\mathbf{T}_m, \forall m,$ by \textbf{Algorithm \ref{alg:penalty}}}
        \STATE{update $\mathbf{D}_m, \forall m,$ by \eqref{opt_solution_D} }
        \STATE{update $\boldsymbol{\mu}_{m,k}$ and $\boldsymbol{\lambda}_{m,k}, \forall m,k,$ by \eqref{opt_solution_mu} and \eqref{opt_solution_lambda}}
        \STATE{update $\mathbf{W}_m, \forall m,$ by \eqref{opt_solution_W} }
        \UNTIL{the fractional increase of the objective value of problem \eqref{problem_se_max_new_3} falls below a threshold}
        \STATE{update penalty factor as $\rho = c \rho$ }
        \UNTIL{the penalty value falls below a threshold}
        \STATE{scale the digital beamformer as $\mathbf{D}_m = \frac{\sqrt{P_t}}{\|\tilde{\mathbf{A}} \tilde{\mathbf{T}}_m \tilde{\mathbf{D}}_m \|_F^2} \tilde{\mathbf{D}}_m$ }
    \end{algorithmic}
\end{algorithm}

\vspace{-0.2cm}
\subsection{Solution for Sub-connected Architectures}
For the sub-connected architecture, the BCD method can also be employed to solve problem \eqref{problem_se_max_new_3}. It can be shown that for fixed $\mathbf{A}$ and $\mathbf{T}_m$, the blocks $\mathbf{D}_m$ and $\mathbf{W}_m$ can be updated using the solutions in \eqref{opt_solution_D} and \eqref{opt_solution_W}, respectively. Therefore, in this subsection, we focus on the subproblem with respect to $\mathbf{A}$ and $\mathbf{T}_m$, which is given by  
\begin{subequations} \label{problem_sub}
    \begin{align}
        \label{obj_1_sub}
        \min_{\mathbf{A}, \mathbf{T}_m}  \quad & \sum_{m=1}^M \left\|\mathbf{W}_m-  \mathbf{A} \mathbf{T}_m \mathbf{D}_m \right\|_F^2 \\
        \mathrm{s.t.} \quad & \eqref{constraint:2}-\eqref{constraint:4}.
    \end{align}
\end{subequations}
Unlike the fully-connected architecture, the analog beamforming matrix $\mathbf{A} \mathbf{T}_m$ in the sub-connected architecture is a block-diagonal matrix rather than a full matrix according to \eqref{ps_sub} and \eqref{ps_sub_2}. By exploiting this spectral structure of matrix $\mathbf{A} \mathbf{T}_m$, problem \eqref{problem_sub} can be directly solved without introducing additional auxiliary variable $\mathbf{V}_m$. In particular, according to \eqref{ps_sub}, the objective function can be equivalently written as
\begin{align} \label{obj_1_sub_2}
    &\eqref{obj_1_sub} = \sum_{n = 1}^{N_{\mathrm{RF}}} \sum_{m=1}^M \left\| \mathbf{\Psi}_{m,n} - \mathbf{A}_{n} e^{-j 2\pi f_m \mathbf{t}_n} \mathbf{p}_{m,n}^H \right\|_F^2 \nonumber \\
    &= \sum_{n=1}^{N_{\mathrm{RF}}} \sum_{m=1}^M \Big( \Xi_{m,n} \mathbf{p}_{m,n}^H \mathbf{p}_{m,n} + \mathrm{tr} \left( \mathbf{\Psi}_{m,n} \mathbf{\Psi}_{m,n}^H \right) \nonumber \\
    &\hspace{1.8cm}- \Re \left\{ \mathbf{p}_{m,n}^H \mathbf{\Psi}_{m,n}^H \mathbf{A}_n e^{-j 2\pi f_m \mathbf{t}_n}  \right\} \Big),
\end{align}
where $\mathbf{\Psi}_{m,n} = \mathbf{W}_m((n-1)N_{\mathrm{sub}}+1:nN_{\mathrm{sub}}, :)$, $\mathbf{p}_{m,n} = \mathbf{D}_m^H(:,n)$, and 
\begin{align}
    \Xi_{m,n} = \sum_{l=1}^{N_{\mathrm{T}}} \mathbf{a}_{n,l}^H \mathbf{a}_{n,l} e^{j 2 \pi f_m t_{n,l}} e^{-j 2 \pi f_m t_{n,l}} = N_{\mathrm{sub}},
\end{align}
where is obtained based on \eqref{ps_sub_2} and the unit-modulus constraint on $\mathbf{a}_{n,l}$. Based on the above result, it can be concluded that only the last term in \eqref{obj_1_sub_2} is related to the optimization variable of interest, which can be further simplified as 
\begin{align} \label{obj_1_sub_3}
    &\sum_{n = 1}^{N_{\mathrm{RF}}} \sum_{m=1}^M \Re \left\{ \mathbf{p}_{m,n}^H \mathbf{\Psi}_{m,n}^H \mathbf{A}_n e^{-j 2\pi f_m \mathbf{t}_n}  \right\} \nonumber\\
    = &\sum_{n = 1}^{N_{\mathrm{RF}}} \sum_{l = 1}^{N_{\mathrm{T}}} \sum_{m=1}^M \Re\left\{ \boldsymbol{\phi}_{m,n,l}^H \mathbf{a}_{n,l} e^{-j 2 \pi f_m t_{n,l}}  \right\},
\end{align} 
where $\boldsymbol{\phi}_{m,n,l} = \mathbf{\Phi}_{m,n}\big((l-1)\frac{N_{\mathrm{sub}}}{N_{\mathrm{T}}}+1:l\frac{N_{\mathrm{sub}}}{N_{\mathrm{T}}}\big)$ and $\mathbf{\Phi}_{m,n} = \mathbf{\Psi}_{m,n} \mathbf{p}_{m,n}$.

\begin{algorithm}[tb]
    \caption{Penalty-FDA Approach for Sub-connected Architectures.}
    \label{alg:penalty_FDA_sub}
    \begin{algorithmic}[1]
        \STATE{initialize $\mathbf{A}$, $\mathbf{T}_m$, $\mathbf{D}_m$, $\mathbf{V}_m$  $\mathbf{W}_m$, $0<c<1$, and $\rho > 0$}
        \REPEAT
        \REPEAT
        \STATE{update $\mathbf{A}$ by \eqref{opt_solution_PS_sub}}
        \STATE{update $\mathbf{T}_m, \forall m,$ by \eqref{opt_solution_TTD_sub}}
        \STATE{update $\mathbf{D}_m, \forall m,$ by \eqref{opt_solution_D} }
        \STATE{update $\boldsymbol{\mu}_{m,k}$ and $\boldsymbol{\lambda}_{m,k}, \forall m,k,$ by \eqref{opt_solution_mu} and \eqref{opt_solution_lambda}}
        \STATE{update $\mathbf{W}_m, \forall m,$ by \eqref{opt_solution_W} }
        \UNTIL{the fractional increase of the objective value of problem \eqref{problem_se_max_new_3} falls below a threshold}
        \STATE{update penalty factor as $\rho = c \rho$ }
        \UNTIL{the penalty value falls below a threshold}
        \STATE{scale the digital beamformer as $\mathbf{D}_m = \frac{\sqrt{P_t}}{\|\tilde{\mathbf{A}} \tilde{\mathbf{T}}_m \tilde{\mathbf{D}}_m \|_F^2} \tilde{\mathbf{D}}_m$ }
    \end{algorithmic}
\end{algorithm}

Given the above derivations, the variables $\mathbf{a}_{n,l}$ and $t_{n,l}$ can be optimized individually to maximize \eqref{obj_1_sub_3} for different values of $n$ and $l$, leading to the following optimization problem:
\begin{subequations}
    \begin{align}
        \max_{\mathbf{a}_{n,l}, t_{n,l}}  \quad & \sum_{m=1}^M \Re\left\{ \boldsymbol{\phi}_{m,n,l}^H \mathbf{a}_{n,l} e^{-j 2 \pi f_m t_{n,l}}  \right\}\\
        \mathrm{s.t.} \quad &  |\mathbf{a}_{n,l}(i)| = 1, \forall i, \\
        & t_{n,l} \in [0, t_{\max}].
    \end{align}
\end{subequations}
For any given $t_{n,l}$, the optimal solution of $\mathbf{a}_{n,l}$ is given by 
\begin{equation} \label{opt_solution_PS_sub}
    \mathbf{a}_{n,l} = e^{j \angle \left(\sum_{m=1}^M \boldsymbol{\phi}_{m,n,l} e^{j 2 \pi f_m t_{n,l}}\right)}.
\end{equation} 
For any given $\mathbf{a}_{n,l}$, the time delay $t_{n,l}$ can be obtained by the low-complex one-dimensional search, resulting in the following solution:
\begin{equation} \label{opt_solution_TTD_sub}
    t_{n,l} = \arg \max_{t_{n,l} \in \mathcal{Q}} \sum_{m=1}^M \Re\left\{ \boldsymbol{\phi}_{m,n,l}^H \mathbf{a}_{n,l} e^{-j 2 \pi f_m t_{n,l}}  \right\}.
\end{equation}     

The FDA approach based on BCD for sub-connected architectures is summarized in \textbf{Algorithm \ref{alg:penalty_FDA_sub}}. The convergence of this algorithm is guaranteed because the objective value is non-increasing in each step of block coordinate descent in the inner loop and is bounded from below. In each iteration, the complexities of updating $\mathbf{A}$ and $\mathbf{T}_m$ are given by $\mathcal{O}(M (N^2K^2 + N) )$ and $\mathcal{O}(M N^2 Q / (N_{\mathrm{RF}} N_{\mathrm{T}}) )$, respectively, while the complexities of updating $\mathbf{D}_m$, $\boldsymbol{\mu}_{m,k}$, $\boldsymbol{\lambda}_{m,k}$, and $\mathbf{W}_m$ are the same as the counterpart in \textbf{Algorithm \ref{alg:penalty_FDA}}. We note that the main complexity of the proposed \textbf{Algorithm \ref{alg:penalty_FDA}} and \textbf{Algorithm \ref{alg:penalty_FDA_sub}} stems from the matrix multiplication and inversion for calculating the closed-form updates over a large number of iterations. To accelerate the algorithm and facilitate its practical implementation, deep unfolding can be employed to transform the iterative algorithm into a trainable deep neural network with a small number of layers, based on the structure of the closed-form updates in each iteration \cite{hu2020iterative, shlezinger2022model}.
\section{Heuristic Two-stage Approach} \label{sec:HTS}

Although the FDA approach directly optimizes the hybrid beamformer and can obtain a stationary point with the aid of the penalty method, it has the following drawbacks.
\begin{itemize}
    \item First, the FDA approach requires the joint optimization of the equivalent fully-digital beamformers and the analog beamformers. In near-field communications with ELAAs, the dimension of these beamformers is typically quite large, which leads to relatively high computational complexity for solving this optimization problem.
    \item Second, the FDA approach requires full knowledge of the CSI. Similarly, due to the large number of antenna elements in near-field communications, the complexity of the channel estimation can be extremely high.
\end{itemize}
To address these issues, in this section, we propose an HTS approach, where the analog and digital beamformers are designed in two separate stages. At the cost of slight performance loss, this approach not only significantly reduces the complexity of optimization, but also simplifies channel estimation with the aid of fast beam training or beam tracking \cite{zhang2022fast, chen2023beam}. In particular, the communication protocol including the proposed HTS approach and beam training is illustrated in Fig. \ref{fig:protocol}. The advantages of this protocol in reducing the computational complexity and signaling overhead of channel estimation have been widely demonstrated in the literature \cite{heath2016overview, tan2021wideband}. Thus, we only focus on the beamforming design in the following.

\vspace{-0.3cm}
\subsection{Heuristic Analog Beamforming Optimization}
The design of the analog beamformer is based on the results of low-complexity near-field beam training or tracking \cite{zhang2022fast, chen2023beam}. The spatial wideband effect can be avoided by utilizing only a small fraction of the bandwidth during beam training or tracking. We note that beam training selects the beam from a predefined codebook that has the maximum received power at the user. Given that LoS channels are much more dominant than NLoS channels, the beam obtained from beam training will be consistent with LoS channels. In other words, the approximate location of the user $(\theta_k, r_k)$ can be obtained by checking the location of the obtained beam in the predefined codebook. Furthermore, beam tracking is an ongoing process that adjusts the beam in response to changes in the environment or the movement of the user, thus avoiding the searching process in beam training. In this paper, we assume that the beam obtained by beam training or tracking perfectly matched the optimal beam to establish a performance upper bound. Without the full CSI, the analog beamformer can be designed based on this location information related to LoS channels. In particular, the received signal of user $k$ on subcarrier $m$ through the LoS channel is given by  
\begin{equation} \label{LoS_only_signal}
    \bar{y}_{m,k} = \beta_{m,k} \mathbf{b}^T(f_m,\theta_k, r_k) \mathbf{A} \mathbf{T}_m \mathbf{D}_m \mathbf{s}_m + n_{m,k}.
\end{equation} 
In the following, we focus on the analog beamforming design for fully-connected architectures, which can be easily extended to cases of sub-connected architectures.

\begin{figure}[t!]
    \centering
    \includegraphics[width=0.45\textwidth]{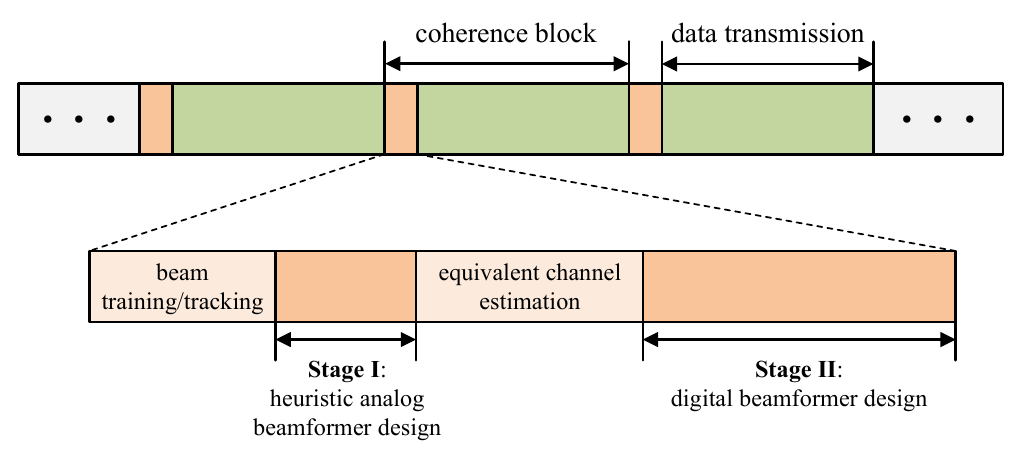}
    \caption{Communication protocol including the proposed HTS approach and beam training/tracking.}
    \label{fig:protocol}
\end{figure}

Given the user locations, the main idea of the heuristic analog beamforming design is to the analog beamformer for each RF chain to maximize the signal power for a specific user through the LoS channel \cite{alkhateeb2015limited}. Specifically, $K$ out of $N_{\mathrm{RF}}$ RF chains is used to serve the $K$ users. For fully-connected architectures, the analog beamformer on subcarrier $m$ for the $k$-th RF chain is given by  
\begin{equation} \label{eqn:hybrid_vector}
    \mathbf{v}_{m,k} = \begin{bmatrix}
        \mathbf{a}_{k,1} e^{-j 2\pi f_m t_{k,1}} \\[-0.4em]
        \vdots \\
        \mathbf{a}_{k,N_{\mathrm{T}}} e^{-j 2\pi f_m t_{k, N_{\mathrm{T}}}}
    \end{bmatrix}.
\end{equation} 
Ideally, given the location information $\theta_k$ and $r_k$ of user $k$, the optimal analog beamformer $\mathbf{v}_{m,k}$ that maximizes the signal power at user $k$ through the LoS channel is given by 
\begin{align} \label{opt_array_beamformer}
    \mathbf{v}_{m,k}^{\mathrm{opt}} = &\arg \max_{\mathbf{v}_{m,k}} \big| \mathbf{b}^T(f_m, \theta_k, r_k) \mathbf{v}_{m,k}  \big| \nonumber \\
    = & \mathbf{b}^*(f_m, \theta_k, r_k) e^{j\psi_{m,k}}, \forall m,
\end{align}
where $\psi_{m,k}$ is an arbitrary phase. Unfortunately, due to the hardware limitation imposed on $\mathbf{v}_{m,k}$, it is generally impossible to find $\mathbf{a}_{k,l}$ and $t_{k,l}, \forall l,$ such that $\mathbf{v}_{m,k} = \mathbf{v}_{m,k}^{\mathrm{opt}}, \forall m$. As a remedy, in the following, we propose two low-complexity analog beamforming designs based on the approximated and exact LoS channels, respectively.

\begin{figure}[t!]
    \centering
    \includegraphics[width=0.3\textwidth]{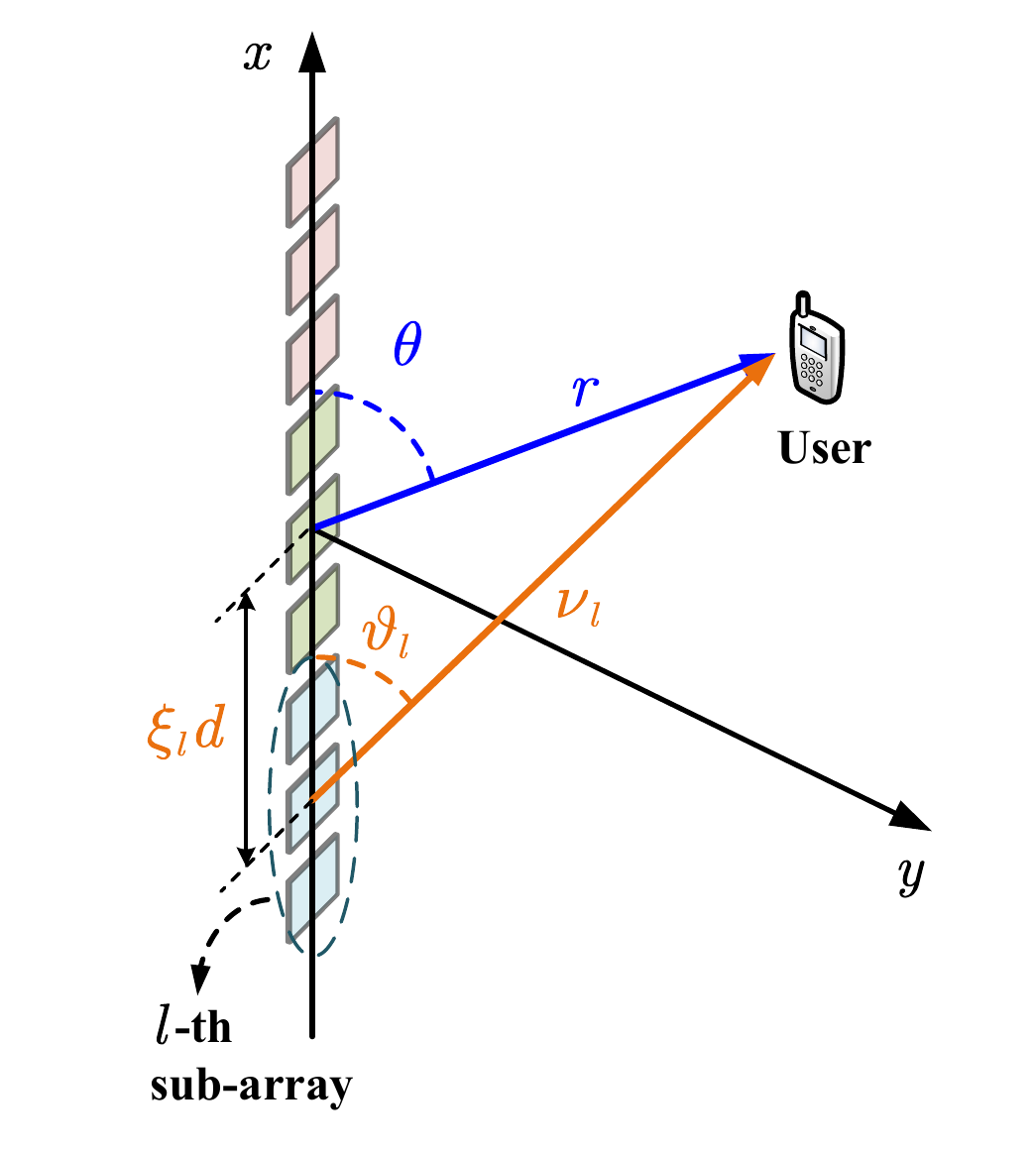}
    \caption{Illustration of the piecewise-near-field channel model.}
    \label{fig:channel_approx}
\end{figure}

\subsubsection{Piecewise-Near-Field (PNF)-based Design} \label{sec_non_robust}
This design employs a PNF approximation of the LoS channels. In practice, the PNF approximation is based on the fact that the entire antenna array is divided into $N_{\mathrm{T}}$ sub-arrays by TTDs. Therefore, we divide the array response vector $\mathbf{b}_N(f_m, \theta, \nu)$ into $N_{\mathrm{T}}$ sub-vectors corresponding to $N_{\mathrm{T}}$ TTDs as follows:
\begin{equation} \label{eq_49}
    \mathbf{b}(f_m, \theta, r) = \begin{bmatrix}
        \boldsymbol{\varphi}_{m,1}(\theta, r) \\
        \vdots\\
        \boldsymbol{\varphi}_{m,N_{\mathrm{T}}}(\theta, r)
    \end{bmatrix},
\end{equation}
where $\boldsymbol{\varphi}_l(f_m, \theta, \nu) \in \mathbb{C}^{\frac{N}{N_{\mathrm{T}}} \times 1}$ denotes a vector consisting of entries from the $((l-1)N_{\mathrm{T}}+1)$-th to the $lN_{\mathrm{T}}$-th row of vector $\mathbf{b}(f_m, \theta, \nu)$.  
Let $\vartheta_{k,l}$ and $\nu_{k,l}$ denote the angle and distance of the user $k$ with respect to the center of the $l$-th sub-array, respectively, and also denote the array response vector of the $l$-th sub-array. According to the geometric relationship illustrated in Fig. \ref{fig:channel_approx}, they can be calculated as follows: 
\begin{align}
    \vartheta_l = &\arccos\left(\frac{r \cos\theta - \xi_l d}{\nu_l}\right), \\
    \label{eq_50}
    \nu_l = &\sqrt{r^2 + \xi_l^2 d^2 - 2 r \xi_l d \cos \theta},
\end{align}  
where $\xi_l = (l - 1 - \frac{N_{\mathrm{T}}-1}{2}) N/N_{\mathrm{T}}$. Then, the propagation distance with respect to the $q$-th element of the $l$-th sub-array can be recalculated as 
\begin{equation} \label{eqn_v_l_q}
    \tilde{\nu}_{l,q} = \sqrt{ \nu_l + \tilde{\chi}_q^2 d^2 - 2 \nu_l \tilde{\chi}_q d  \cos \vartheta_l },
\end{equation}
where $\tilde{\chi}_q = q -1 - \frac{N/N_{\mathrm{T}}-1}{2}$. Therefore, the $n$-th entry of the array response vector $\boldsymbol{\varphi}_{m,l}(\theta, r)$, denoted by $\varphi_{m,l,q}$, can be rewritten as follows:
\begin{align} \label{b_vec_1}
    &\varphi_{m,l,q} = e^{-j \frac{2 \pi f_m}{c} (\tilde{\nu}_{l,q} - r) } = e^{-j \frac{2 \pi f_m}{c} (\nu_l-r) } e^{-j \frac{2 \pi f_m}{c} (\tilde{\nu}_{l,q} - \nu_l) }.
\end{align}
It can be observed that $\varphi_{m,l,q}$ is now divided into two components: $e^{-j \frac{2\pi f_m}{c} (\nu_l - r)}$ and $e^{-j \frac{2\pi f_m}{c} (\tilde{\nu}_{l,q} - \nu_l) }$.
The first component represents the phase shift attributable to the variance in propagation delay between the center of the sub-array and the center of the entire array, while the second component represents the phase shift arising from the difference in propagation delay within the sub-array. In order to maximize the array gain at all subcarriers, the time delay of TTDs should be designed to compensate for the difference in these propagation delays. However, the propagation delay in the second component is specific to each antenna element, making it difficult to address by a limited number of TTDs. To solve this issue, we approximate $\varphi_{m,l,q}$ as follows:
\begin{align} \label{approx_phi}
    \varphi_{m,l,q} &= e^{-j \frac{2 \pi f_m}{c} (\nu_l-r) } e^{-j \frac{2 \pi f_c}{c} (\tilde{\nu}_{l,q} - \nu_l) } e^{-j \frac{2 \pi (f_m - f_c)}{c} (\tilde{\nu}_{l,q} - \nu_l) } \nonumber \\
        &\approx e^{-j \frac{2 \pi f_m}{c} (\nu_l-r) } e^{-j \frac{2 \pi f_c}{c} (\tilde{\nu}_{l,q} - \nu_l) }.
\end{align}
In this approximation, the term $e^{-j \frac{2 \pi (f_m - f_c)}{c} (\tilde{\nu}_{l,q} - \nu_l) }$ is omitted, which is valid when the value of $(\tilde{\nu}_{l,q} - \nu_l)$ is sufficiently small, rendering the product $(f_m - f_c)(\tilde{\nu}_{l,q} - \nu_l)$ negligible. In particular, $(f_m - f_c)$ is determined by the system bandwidth, whereas $(\tilde{\nu}_{l,q} - \nu_l)$ is influenced by the size of sub-arrays. Consequently, reducing $(\tilde{\nu}_{l,q} - \nu_l)$ necessitates smaller sub-arrays, which in turn increases the requirement for more TTDs. Based on the approximation in \eqref{approx_phi}, the array response vector $\mathbf{b}(f_m, \theta, r)$ can be approximated as 
\begin{align} \label{eqn_approx_response}
    \hat{\mathbf{b}}(f_m, \theta, r) = \begin{bmatrix}
        \hat{\boldsymbol{\varphi}}_1(\theta, r) e^{-j \frac{2 \pi f_m}{c} (\nu_1-r) } \\
        \vdots\\
        \hat{\boldsymbol{\varphi}}_{N_{\mathrm{T}}}(\theta,r) e^{-j \frac{2 \pi f_m}{c} (\nu_{N_{\mathrm{T}}}-r) }
    \end{bmatrix},
\end{align}
where $\hat{\boldsymbol{\varphi}}_l(\theta, r) \in \mathbb{C}^{\frac{N}{N_{\mathrm{T}}} \times 1}$ is given by 
\begin{equation}
    \hat{\boldsymbol{\varphi}}_l(\theta, r) = \left[e^{-j \frac{2 \pi f_c}{c} (\tilde{\nu}_{l,1} - \nu_l)},\dots, e^{-j \frac{2 \pi f_c}{c} (\tilde{\nu}_{l,N/N_{\mathrm{T}}} - \nu_l)} \right]^T.
\end{equation} 
The relationship between the worst-case accuracy of the PNF approximation and the number of TTDs is given in the following proposition.

\begin{figure}[t!]
    \centering
    \includegraphics[width=0.45\textwidth]{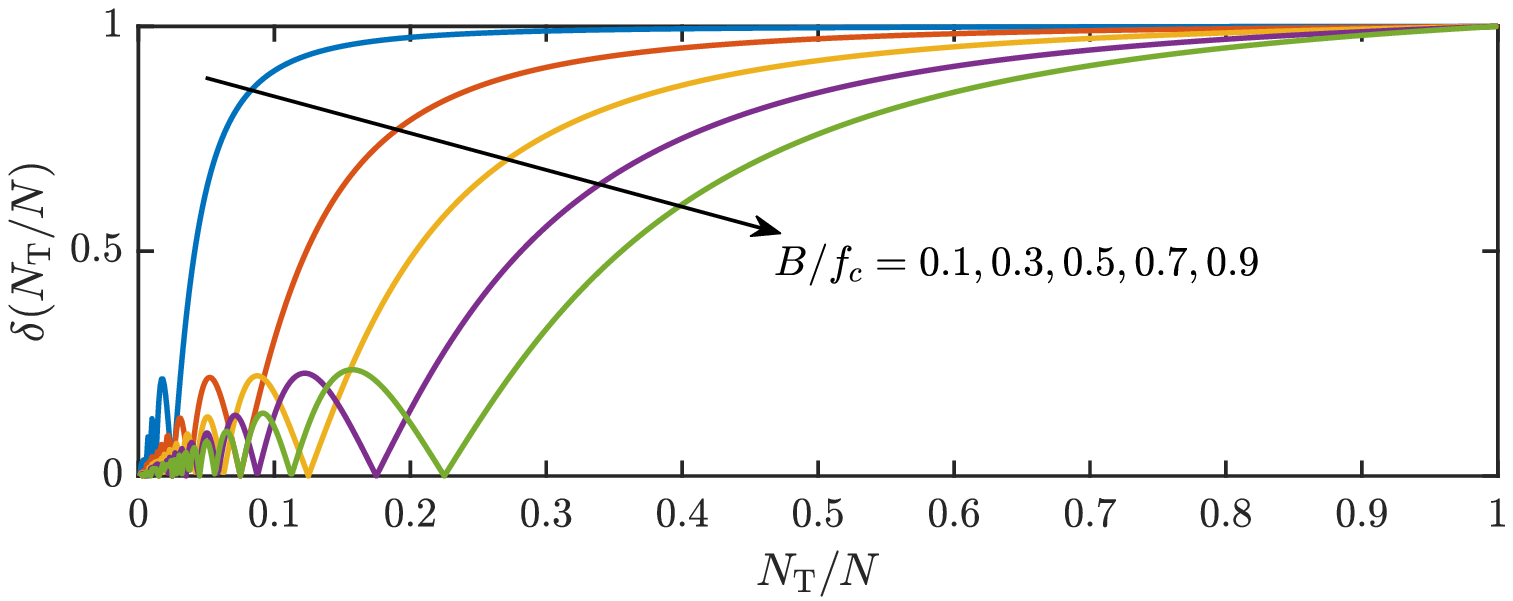}
    \caption{Illustration of function $\delta\big(\frac{N_{\mathrm{T}}}{N}, \frac{B}{f_c}\big)$.}
    \label{fig:error_function}
\end{figure}

\begin{proposition} \label{proposition_1}
    \normalfont Define the normalized array gain realized by $\hat{\mathbf{b}}(f_m, \theta, r)$ at the location $(\theta, r)$ as follows:
    \begin{equation}
        G(f_m, \theta, r) = \frac{1}{N} \left| \mathbf{b}^H(f_m, \theta, r) \hat{\mathbf{b}}(f_m, \theta, r) \right|.
    \end{equation}
    Then, given a positive threshold $\Delta$ slightly less than 1, to ensure that $\min_{f_m, \theta, r} G(f_m, \theta, r) \ge \Delta$, the following condition should be satisfied:
    \begin{equation} \label{threshold}
        \delta\left(\frac{N_{\mathrm{T}}}{N}, \frac{B}{f_c}\right) \triangleq \left| \frac{N_{\mathrm{T}}}{N} \frac{\sin \left( \frac{\pi B N}{4 f_c N_{\mathrm{T}}}  \right)}{\sin \left(\frac{\pi B}{4 f_c} \right)} \right| \ge \Delta.
    \end{equation}
\end{proposition}

\begin{IEEEproof}
    Please refer to Appendix \ref{proposition_1_proof}.
\end{IEEEproof}

According to \textbf{Proposition \ref{proposition_1}}, the accuracy of the PNF approximation in the worst-case scenario depends on two key ratios: $N_{\mathrm{T}}/N$ and $B/f_c$. The behavior of the function $\delta\big(\frac{N_{\mathrm{T}}}{N}, \frac{B}{f_c}\big)$ in relation to these ratios is depicted in Fig. \ref{fig:error_function}. It is evident that an increase in bandwidth $B$ necessitates a corresponding increase in the number of TTDs to maintain a specified accuracy threshold $\Delta$. For instance, consider a scenario with $\Delta=0.8$, $N=512$, and $f_c = 100$ GHz. In this context, as per the criterion given in \eqref{threshold}, a bandwidth of $B = 10$ GHz requires $N_{\mathrm{T}} \ge 36$ TTDs. Conversely, for a larger bandwidth of $B = 20$ GHz, the requirement escalates to $N_{\mathrm{T}} \ge 105$ TTDs. It is important to note that the above analysis pertains to the worst-case performance. When considering the average performance, the number of TTDs required can be significantly lower, which will be demonstrated by the numerical results presented in Section \ref{sec:result}.

Based on the approximated array response vector in \eqref{eqn_approx_response}, the average array gain realized by the hybrid analog beamformer $\mathbf{v}_{m,k}$ for user $k$ at all subcarriers can be approximated as 
\begin{align} \label{G_func}
    &\bar{G}(\theta_k, r_k) = \frac{1}{M} \sum_{m=1}^M \big| \hat{\mathbf{b}}^T(f_m, \theta_k, r_k) \mathbf{v}_{m,k}  \big| \nonumber \\
    &= \frac{1}{M} \sum_{m=1}^M \left|\sum_{l=1}^{N_{\mathrm{T}}} \hat{\boldsymbol{\varphi}}_l^T(\theta_k, r_k) \mathbf{a}_{k,l} e^{-j 2\pi f_m \left( \frac{\nu_l(\theta_k, r_k) - r_k}{c} + t_{k,l}  \right)} \right|.
\end{align}
Then, the analog beamforming design problem for the $k$-th RF chain can be formulated as 
\begin{subequations} \label{HTS_opt}
    \begin{align}
        \max_{\mathbf{a}_{k,l}, t_{k,l}}  \quad & \bar{G}(\theta_k, r_k) \\
        \mathrm{s.t.} \quad &  |\mathbf{a}_{k,l}(i)| = 1, \forall l, i, \\
        & t_{k,l} \in [0, t_{\max}], \forall l.
    \end{align}
\end{subequations}
The above optimization problem is essentially a phase alignment problem. According to \eqref{G_func}, to maximize $\bar{G}(\theta_k, r_k)$, the PS-based analog beamformer $\mathbf{a}_{k,l}$ can be designed to match the phase of the frequency-independent vector $\hat{\boldsymbol{\varphi}}_l^T(\theta_k, r_k)$, resulting the following solution:
\begin{equation} \label{opt_PS} 
    \mathbf{a}_{k,l} = \hat{\boldsymbol{\varphi}}_l^*(\theta_k, r_k).
\end{equation}  
Given the above solution, problem \eqref{HTS_opt} reduces to 
\begin{subequations} \label{HTS_opt_2}
    \begin{align}
        \max_{t_{k,l}}  \quad & \sum_{m=1}^M \left| \sum_{l=1}^{N_{\mathrm{T}}} e^{-j 2\pi f_m \left( \frac{\nu_l(\theta_k, r_k) - r_k}{c} + t_{k,l}  \right)} \right| \\
        \mathrm{s.t.} \quad & t_{k,l} \in [0, t_{\max}], \forall l.
    \end{align}
\end{subequations}
For the solution to this optimization problem, we have the following proposition.
\begin{proposition} \label{proposition_2}
    \normalfont If $t_{\max} \ge \frac{N (N_{\mathrm{T}} - 1)d}{N_{\mathrm{T}} c}$, the optimal solution to problem \eqref{HTS_opt_2} is given by 
    \begin{equation}
        t_{k,l} = \frac{\nu_{\max,k} - \nu_l(\theta_k, r_k)}{c},
    \end{equation}
    where $\nu_{\max,k} = \max \{\nu_1(\theta_k, r_k),\dots, \nu_{N_{\mathrm{T}}}(\theta_k, r_k) \}$. 
\end{proposition}
\begin{IEEEproof}
    Please refer to Appendix \ref{proposition_2_proof}. 
\end{IEEEproof}

In \textbf{Proposition \ref{proposition_2}}, we prove that the optimal TTD coefficients to problem \eqref{HTS_opt_2} have the closed-form solution when $t_{\max}$ is sufficiently large. When $t_{\max}$ is small, problem \eqref{HTS_opt_2} can be solved through an iterative one-dimensional search. More specifically, in each iteration, each time delay $t_{k,l}$ is updated by one-dimensional search as given in \eqref{opt_solution_TTD} and \eqref{opt_solution_TTD_sub} by fixing the other time delays. The complexity of the one-dimensional search in each iteration is $\mathcal{O}(QM N_{\mathrm{T}})$.

\subsubsection{Robust Design} \label{sec_robust}
Although the idea of the PNF-based design can facilitate the closed-form solution of the analog beamformers, it might result in substantial performance degradation under certain conditions. Specifically, this design relies extensively on the accuracy of the PNF approximation of the channel responses. However, this accuracy can diminish significantly when the number of TTDs $N_{\mathrm{T}}$ is small, or when the system operates over an ultra-large bandwidth $B$, as illustrated in Fig. \ref{fig:error_function}. To address this issue, in this section, we further propose a robust design that jointly optimizes the PS-based and TTD-based analog beamformers based on the exact LoS channel. The exact average array gain achieved by the overall analog beamformer for user $k$ is given by 
\begin{align}
    \tilde{G}(\theta_k, r_k) = &\frac{1}{M} \sum_{m=1}^M \left| \mathbf{b}^T(f_m, \theta_k, r_k) \mathbf{v}_{m,k}  \right| \nonumber \\
    = & \frac{1}{M} \sum_{m=1}^M \left|\sum_{l=1}^{N_{\mathrm{T}}} \boldsymbol{\varphi}_{m,l}^T(\theta_k, r_k) \mathbf{a}_{k,l} e^{-j 2\pi f_m t_{k,l}} \right|,
\end{align}
where $\boldsymbol{\varphi}_{m,l}$ is defined in \eqref{eq_49}. In the sequel, we jointly optimize the PS-based and TTD-based analog beamformers to maximize the exact average array gain, resulting in the following optimization problem:
\begin{subequations} \label{HTS_opt_3}
    \begin{align}
        \label{HTS_opt_3_obj}
        \max_{\mathbf{a}_{k,l}, t_{k,l}}  \quad & \sum_{m=1}^M \left|\sum_{l=1}^{N_{\mathrm{T}}} \boldsymbol{\varphi}_{m,l}^T(\theta_k, r_k) \mathbf{a}_{k,l} e^{-j 2\pi f_m t_{k,l}} \right| \\
        \mathrm{s.t.} \quad &  |\mathbf{a}_{k,l}(i)| = 1, \forall l, i, \\
        & t_{k,l} \in [0, t_{\max}], \forall l.
    \end{align}
\end{subequations}
This problem can also be solved using BCD given as follows.

When $t_{k,l}$ is fixed, problem \eqref{HTS_opt_3} can be reformulated as 
\begin{subequations} \label{HTS_opt_4}
    \begin{align}
        \max_{\bar{\mathbf{a}}_k}  \quad & \sum_{m=1}^M \left|\boldsymbol{\eta}_{m,k}^H \bar{\mathbf{a}}_k\right|\\
        \mathrm{s.t.} \quad & |\bar{\mathbf{a}}_{k}(n)| = 1, \forall n,
    \end{align}
\end{subequations}
where $\bar{\mathbf{a}}_k = [ \mathbf{a}_{k,1}^T,\dots,\mathbf{a}_{k, N_{\mathrm{T}}}^T ]^T$ and
\begin{equation}
    \boldsymbol{\eta}_{m,k} = \begin{bmatrix}
        \boldsymbol{\varphi}_{m,1}^*(\theta_k, r_k) e^{j 2\pi f_m t_{k,1}} \\
        \vdots\\
        \boldsymbol{\varphi}_{m,N_{\mathrm{T}}}^*(\theta_k, r_k) e^{j 2\pi f_m t_{k,N_{\mathrm{T}}}}
    \end{bmatrix}.
\end{equation}
This optimization problem is non-convex. We propose a low-complexity iterative algorithm to solve it based on the majorization-minimization (MM) method \cite{sun2016majorization}. In particular, we first construct the following surrogate function of the objective function:
\begin{align}
f \left( \bar{\mathbf{a}}_k | \bar{\mathbf{a}}_k^{(t)} \right) \triangleq \sum_{m=1}^M \Re\left\{\frac{ \left(\bar{\mathbf{a}}_k^{(t)}\right)^H \boldsymbol{\eta}_{m,k} \boldsymbol{\eta}_{m,k}^H \bar{\mathbf{a}}_k  }{\left| \boldsymbol{\eta}_{m,k}^H \bar{\mathbf{a}}_k^{(t)} \right|} \right\},
\end{align}
where $\bar{\mathbf{a}}_k^{(t)}$ is the solution of $\bar{\mathbf{a}}_k$ obtained in the $t$-th iteration of the MM algorithm. According to the Cauchy-Schwarz inequality, we have 
\vspace{-0.2cm}
\begin{equation}
    f \left( \bar{\mathbf{a}}_k | \bar{\mathbf{a}}_k^{(t)} \right) \le \sum_{m=1}^M \left|\boldsymbol{\eta}_{m,k}^H \bar{\mathbf{a}}_k\right|,
    \vspace{-0.1cm}
\end{equation}
where the equality is achieved at $\bar{\mathbf{a}}_k^{(t)} = \bar{\mathbf{a}}_k$. With the above surrogate function $f \left( \bar{\mathbf{a}}_k | \bar{\mathbf{a}}_k^{(t)} \right)$, problem \eqref{HTS_opt_4} can be solved by the following process \cite{sun2016majorization}:
\begin{equation}
    \bar{\mathbf{a}}_k^{(t+1)} = \argmax_{|\bar{\mathbf{a}}_{k}(n)| = 1, \forall n} f \left( \bar{\mathbf{a}}_k | \bar{\mathbf{a}}_k^{(t)} \right) = e^{j \angle \left(\mathbf{q}_k^{(t)}  \right)},
\end{equation}
where
\begin{equation}
    \mathbf{q}_k^{(t)} = \sum_{m=1}^M \frac{\boldsymbol{\eta}_{m,k} \boldsymbol{\eta}_{m,k}^H \bar{\mathbf{a}}_k^{(t)}}{ \left| \boldsymbol{\eta}_{m,k}^H \bar{\mathbf{a}}_k^{(t)} \right| }.
\end{equation}
After obtaining $\bar{\mathbf{a}}_k$ through the iterative MM algorithm, the PS-based analog beamformer can be designed as
$\mathbf{a}_{k,l} = \bar{\mathbf{a}}_k\big( (l-1) \frac{N}{N_{\mathrm{T}}} + 1 : l \frac{N}{N_{\mathrm{T}}} \big), \forall l$.

When $\mathbf{a}_{k,l}$ is fixed, problem \eqref{HTS_opt_3} reduces to
\begin{subequations} \label{HTS_opt_5}
    \begin{align}
        \max_{t_{k,l}}  \quad & \sum_{m=1}^M \left| \sum_{l=1}^{N_{\mathrm{T}}} \gamma_{m,k,l}  e^{-j 2\pi f_m  t_{k,l} } \right| \\
        \mathrm{s.t.} \quad & t_{k,l} \in [0, t_{\max}], \forall l,
    \end{align}
\end{subequations}
where $\gamma_{m,k,l} = \boldsymbol{\varphi}_{m,l}^T(\theta_k, r_k) \mathbf{a}_{k,l}$. Similar to problem \eqref{HTS_opt_2}, this problem can also be solved through the iterative one-dimensional search. Then, the solution to problem \eqref{HTS_opt_3} can be obtained by updating $\mathbf{a}_{k,l}$ and $t_{k,l}$ alternately based on the above methods. The complexity for updating $\mathbf{a}_{k,l}, \forall l,$ in each iteration of the MM algorithm is $\mathcal{O}(2 MN^2 + MN)$ and for updating $t_{k,l}, \forall l$, through the one-dimensional search is $\mathcal{O}(QMN_{\mathrm{T}})$. Compared to the PNF-based design, the robust design maximizes the array gain based on the exact channel model, resulting in a better performance but at the cost of incurring higher complexity.

\subsection{Digital Beamforming Optimization}
Given the analog beamformer designed in the previous subsection, the received signal in \eqref{eqn:received_signal} can be concisely reformulated as
\begin{equation}
    y_{m,k} = \tilde{\mathbf{h}}_{m,k}^H \mathbf{D}_m \mathbf{s}_m + n_{m,k},
\end{equation}
where $\tilde{\mathbf{h}}_{m,k} = \mathbf{T}_m^H \mathbf{A}^H \mathbf{h}_{m,k}$ represents the equivalent channel. This equivalent channel, which has a low dimensionality of $N_{\mathrm{RF}}$, can be efficiently estimated using standard channel estimation methods with low complexity. The digital beamformer can then be optimized to maximize the spectral efficiency of the equivalent channel, resulting in the following problem:
\begin{subequations} \label{problem_HTS_D}
    \begin{align}
        \max_{\mathbf{D}_m}  \quad &\sum_{m=1}^M \sum_{k=1}^K \breve{R}_{m,k}(\mathbf{D}_m) \\
        \mathrm{s.t.} \quad & \|\mathbf{A} \mathbf{T}_m \mathbf{D}_m \|_F^2 \le P_t, \forall m,
    \end{align}
\end{subequations}
where
\begin{equation}
    \breve{R}_{m,k} = \log_2 \left( 1 + \frac{| \tilde{\mathbf{h}}_{m,k}^H \mathbf{d}_{m,k} |^2}{\sum_{i \in \mathcal{K}, i \neq k} | \tilde{\mathbf{h}}_{m,k}^H \mathbf{d}_{m,i} |^2 + \sigma_{m,k}^2} \right).
\end{equation}
This problem is a classical sum-rate maximization, which can be effectively solved by the existing method, such as fractional programming \cite{shen2018fractional2}, MM \cite{sun2016majorization}, and weighted minimum mean-squared error \cite{christensen2008weighted} methods. Thus, we omit the details here. Due to the low dimensionality of $\mathbf{D}_m$, the corresponding algorithm has low complexity. Based on the above discussion, the overall HTS approach is summarized in \textbf{Algorithm \ref{alg_heuristic}}. 

For sub-connected architectures, the HTS approach follows a similar procedure as in \textbf{Algorithm \ref{alg_heuristic}}. The difference is that the design of the analog beamformer needs to be based on the array response of the sub-arrays connected to each RF chain.

\begin{algorithm}[tb]
    \caption{Low-Complexity HTS Approach for Fully-connected Architectures.}
    \label{alg_heuristic}
    \begin{algorithmic}[1]
        \STATE{\textbf{Stage I: Heuristic analog beamforming design}}
        \STATE{estimate the user location information $\theta_k$ and $r_k, \forall k$.}
        \FOR{ $n \in \{1,\dots,N_{\mathrm{RF}}\}$  }
        \STATE{If $n \le K$, design $\mathbf{a}_{n,l}$ and $t_{n,l}, \forall l$ based on the location of user $n$, i.e., $\theta_n$ and $r_n$, using the PNF-based method or the robust method. Otherwise, design them based on the location of a randomly selected user.}
        \ENDFOR
        \STATE{\textbf{Stage II: Digital beamforming optimization}}
        \STATE{estimate the equivalent channel $\tilde{\mathbf{h}}_{m,k}, \forall m,k$.}
        \STATE{optimize the digital beamformer $\mathbf{D}_m, \forall m,$ by solving \eqref{problem_HTS_D}.}
    \end{algorithmic}
\end{algorithm}

\section{Numerical Results} \label{sec:result}
In this section, numerical results obtained through Monte Carlo simulations are provided to verify the effectiveness of the proposed FAD and HTS approaches. The following simulation setup is exploited throughout our simulations unless otherwise specified. It is assumed that the BS is equipped with $N=512$ antennas with half-wavelength spacing,  $N_\mathrm{RF} = 4$ RF chains, and $N_{\mathrm{T}} = 16$ TTDs for each RF chain. The maximum time delay of each TTD is set to $t_{\max} = N/(2f_c) = 2.56$ nanosecond (ns) \cite{dai2022delay}. The central frequency, system bandwidth, number of subcarriers, and length of the cyclic prefix are assumed to be $f_c = 100$ GHz, $B = 10$ GHz, $M = 10$, and $L_{\mathrm{CP}} = 4$, respectively. There are $K = 4$ near-field communication users randomly located in close proximity to the BS from $5$ m to $15$ m. The number of NLoS channels between the BS and each user is assumed to be $L_k = 4, \forall k$. The average reflection coefficient of NLoS channels is set to $-15$ dB. The noise power density is set to $-174$ dBm/Hz. 

For the proposed algorithms, the convergence threshold is set to $10^{-3}$. The penalty factors are set to $\rho = \bar{\rho} = 10^3$, and their reduction factor is set to $c = 0.5$. For the one-dimensional search, we set $Q = 10^3$. Furthermore, “HTS-PNF” represents the HTS approach based on the PNF-based analog beamforming design, while “HTS-Robust” refers to the HTS approach using the robust design. Since the performance of the FDA approach might be sensitive to the initialized parameters, the optimization variables of these two algorithms are initialized based on the HTS-PNF approach to guarantee a good performance. All following results are obtained by averaging over $100$ random channel realizations. For performance comparison, we mainly consider the following two benchmark schemes: 1) \textbf{Optimal Digital beamforming (BF):} In this benchmark, each antenna is connected to a dedicated RF chain, which realizes the full-dimensional frequency-dependent beamforming. This benchmark provides an upper-bound performance. The fully-digital beamformers can be optimized using the methods in \cite{shen2018fractional2, sun2016majorization, christensen2008weighted}; 2) \textbf{Conventional BF:} In this benchmark, the analog beamforming is achieved by only frequency-independent PSs. The corresponding beamforming optimization problem can be solved through the proposed penalty-FDA algorithm by setting $t_{n,l} = 0, \forall n, l$.

\begin{figure}[t!]
    \centering
    \includegraphics[width=0.23\textwidth]{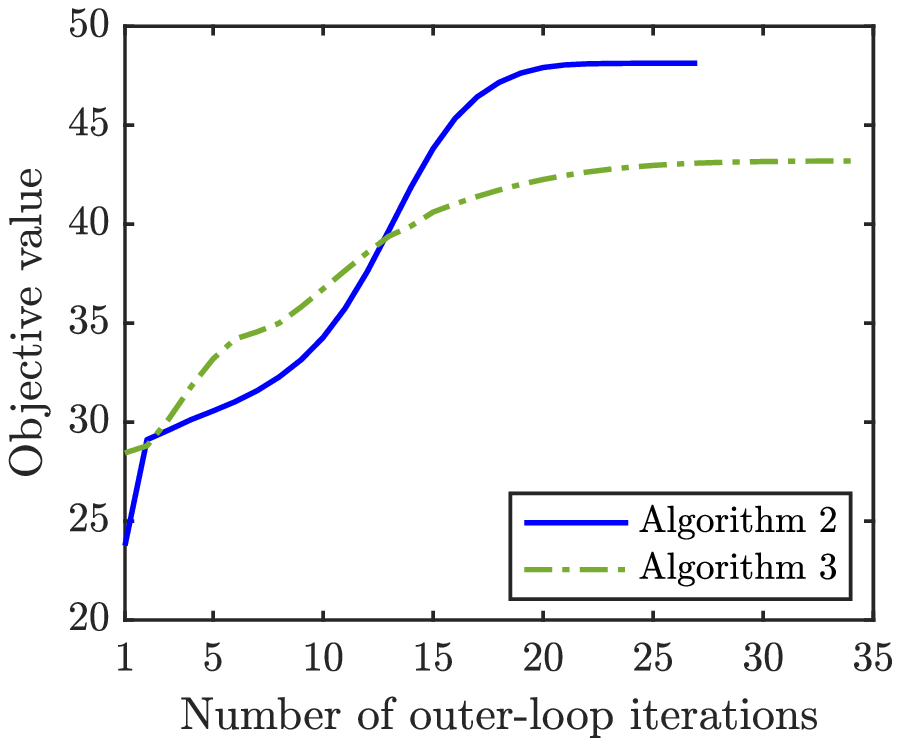}
    \includegraphics[width=0.23\textwidth]{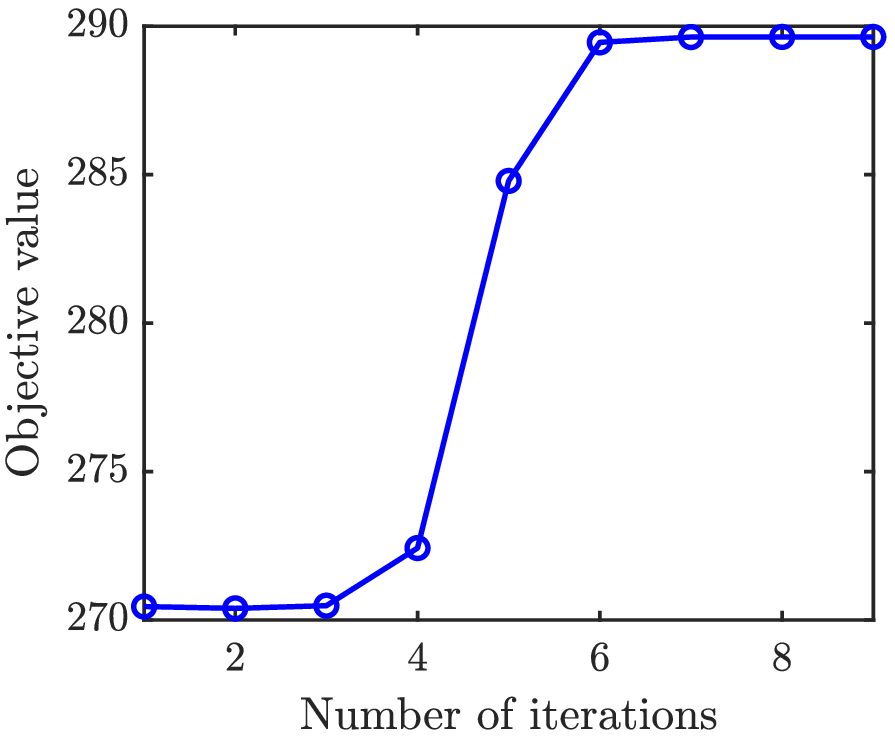}
    \caption{Convergence behavior of the proposed FDA approaches (left) and robust design for analog beamforming optimization (right).}
    \label{fig_convergence}
    \vspace{-0.1cm}
\end{figure}

\subsection{Convergence and Complexity of Proposed Approaches}
Fig. \ref{fig_convergence} demonstrates the convergence behavior of the proposed iterative algorithms. To quantitatively analyze the complexity, we examine the CPU time of each approach using MATLAB R2024a on an Apple M3 silicon chip. For the fully-connected architecture, the CPU times are 37.91 s for FDA, 0.53 s for HTS-Robust, and 0.35 s for HTS-PNF. For the sub-connected architecture, the CPU times are 18.52 s, 0.42 s, and 0.37 s, respectively. The HTS approach significantly reduces computational complexity by avoiding joint optimization of analog and digital beamformers. Additionally, the FDA approach for sub-connected architecture consumes less CPU time than the fully-connected architecture due to the simplified optimization process enabled by the spatial structure of the sub-connected analog beamformers.

\vspace{-0.3cm}
\subsection{Evaluation of Proposed Heuristic Analog Beamforming}
In Fig. \ref{fig:array_gain_bandwidth}, we evaluate the normalized array gain achieved by the proposed heuristic analog beamforming method across different frequencies and bandwidths. Here, we consider the fully-connected architecture and a single user located at $(45^\circ, 10 \text{ m})$ as an example. We compare our approach against the following wideband analog beamforming designs for conventional BF architecture:  1) Mean channel covariance matrix (MCCM) method \cite{chen2020hybrid}, which designs the analog beamformer based on the averaged covariance matrix of array response vectors over all subcarriers; 2) Mean channel matrix (MCM) method \cite{chen2020hybrid}, which designs the analog beamformer based on the averaged array response vector over all subcarriers; and 3) Central frequency (CF) method, where the design is based on the array response at the central frequency. It can be observed that the proposed heuristic method employing TTDs significantly enhances the array gain throughout the entire frequency band. However, with bandwidth expansion from $10$ GHz to $30$ GHz, the proposed PNF-based method maintains high array gain only over a limited range of the frequency band. In contrast, the robust design consistently achieves high array gain over the entire band, confirming its robustness.

\begin{figure}[t!]
    \centering
    \includegraphics[width=0.4\textwidth]{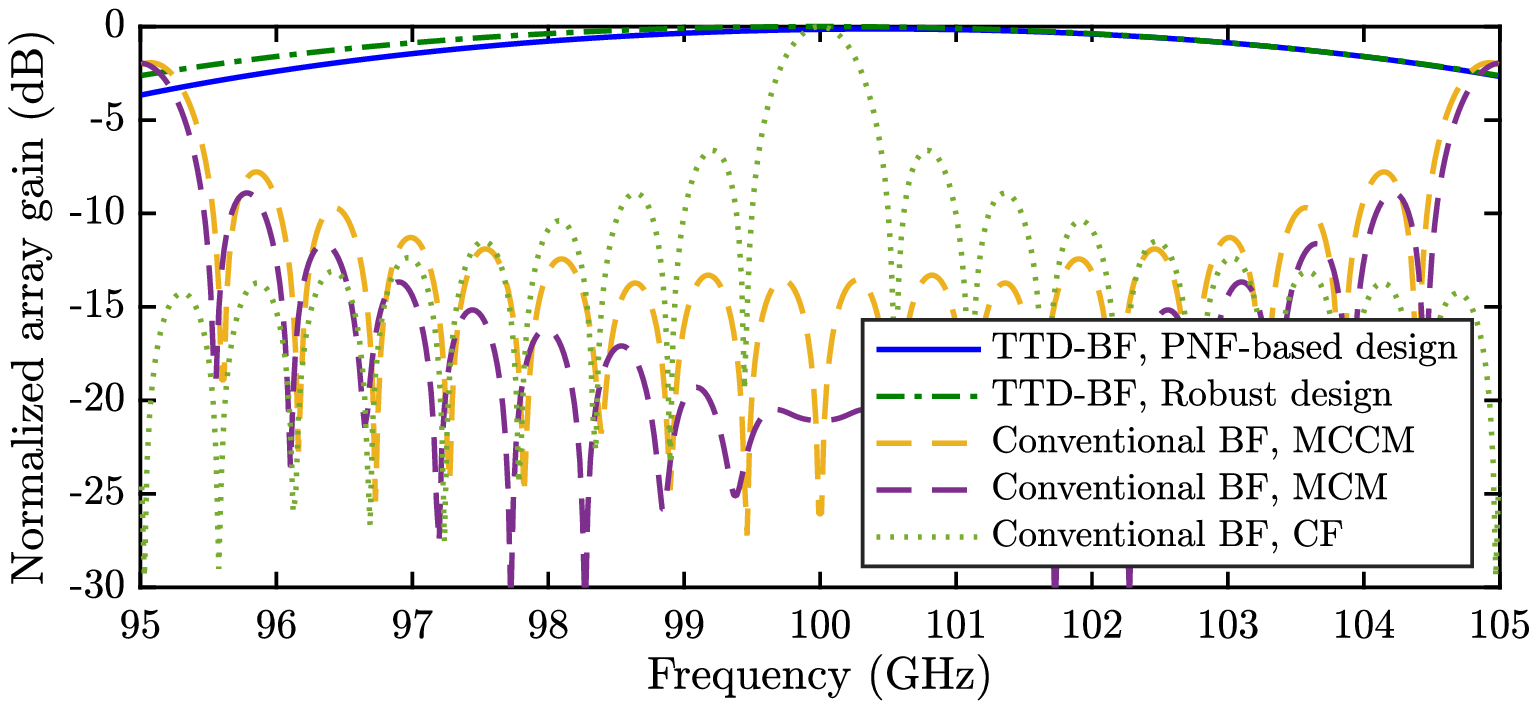}
    \includegraphics[width=0.4\textwidth]{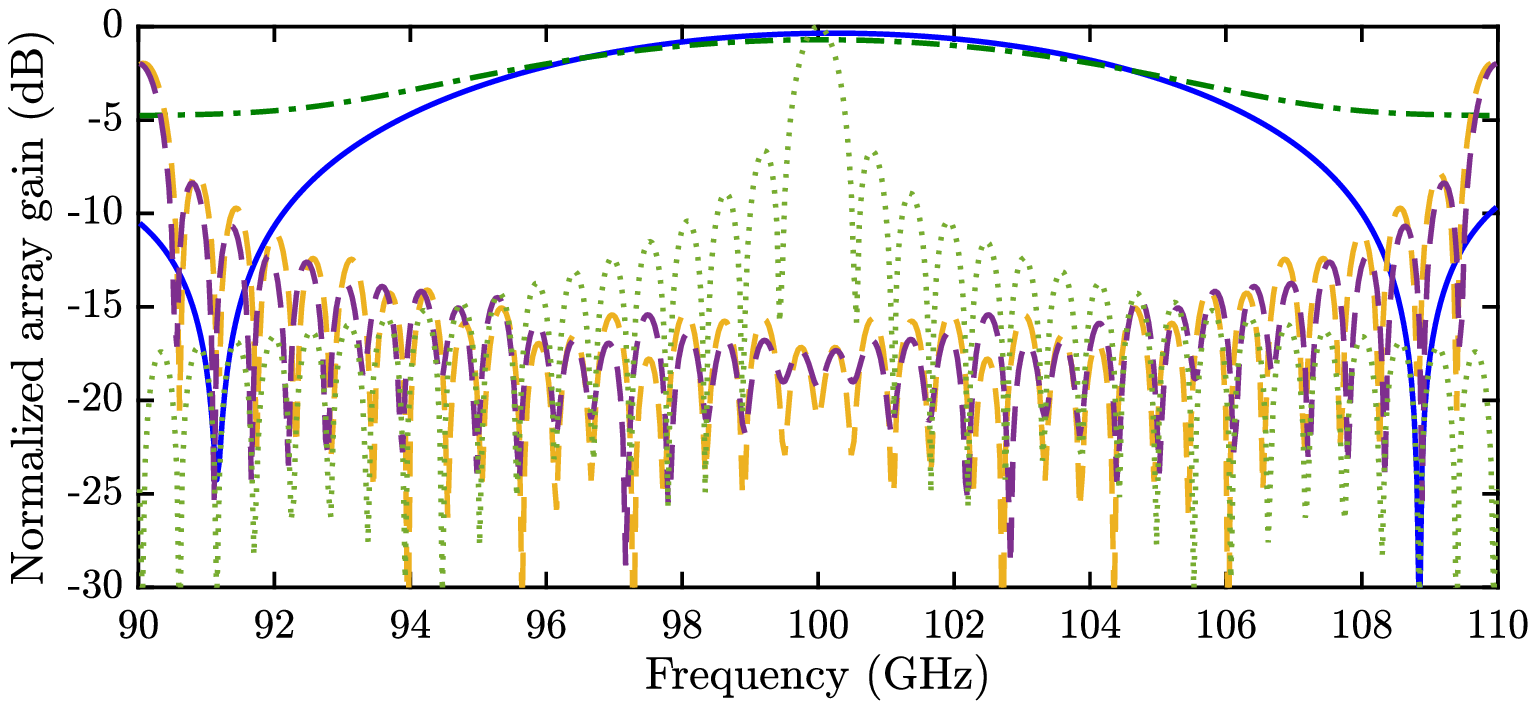}
    \includegraphics[width=0.4\textwidth]{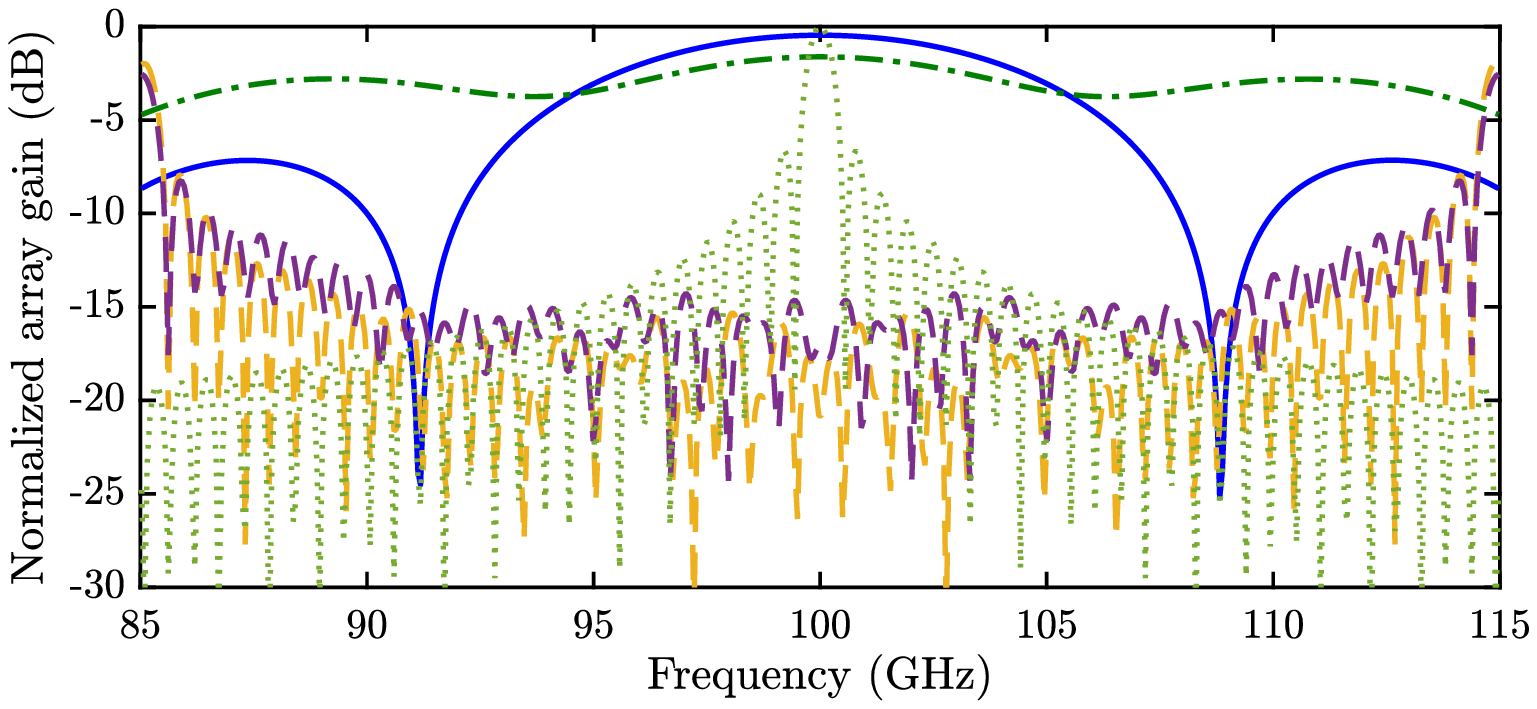}
    \caption{Normalized array gain across the frequency band with $10$ GHz (top), $20$ GHz (middle), and $30$ GHz (bottom) bandwidths.}
    \label{fig:array_gain_bandwidth}
\end{figure} 

In Fig. \ref{fig:array_gain}, we further evaluate the proposed heuristic analog beamforming method with different maximum time delays. In addition to the above benchmarks, we also consider two existing wideband analog beamforming methods using TTD for performance comparison, namely the near-field phase-delayed focusing (PDF) method \cite{cui2021near} and the far-field delayed phase precoding (DPP) method \cite{gao2021wideband, dai2022delay}. The near-field PDF scheme is based on a piecewise-far-field approximation, which approximates the array response of each sub-array as a far-field channel, while the far-field DPP scheme approximates the array response of the whole array through a far-field channel. It can be seen from Fig. \ref{fig:array_gain} that the proposed method achieves significant performance gains compared to the existing method due to the use of a more accurate approximation and the optimization of the maximum delay constraint.

\vspace{-0.3cm}
\subsection{Spectral Efficiency Versus Transmit Power}
In Fig. \ref{fig:rate_vs_power}, we evaluate the spectral efficiency achieved by the different schemes with respect to the transmit power $P_t$. It can be seen that the fully-connected TTD-BF architecture performs close to the optimal digital BF architecture by employing a limited number of TTDs. For the subconnected TTD-BF architecture, although only one sub-array per RF chain can be used to transmit the signal, the resulting performance loss is not significant. In addition, the FDA approach is the most efficient approach in both fully-connected and sub-connected architectures due to its ability to achieve a stationary point solution. Although only partial CSI is used to heuristically design the analog beamformer, the performance loss of the HTS approaches is negligible compared to the FDA approach. On the contrary, conventional BF architectures utilizing only frequency-dependent PSs suffer from significant performance degradation due to significant spatial wideband effects.

\begin{figure}[t!]
    \centering
    \includegraphics[width=0.4\textwidth]{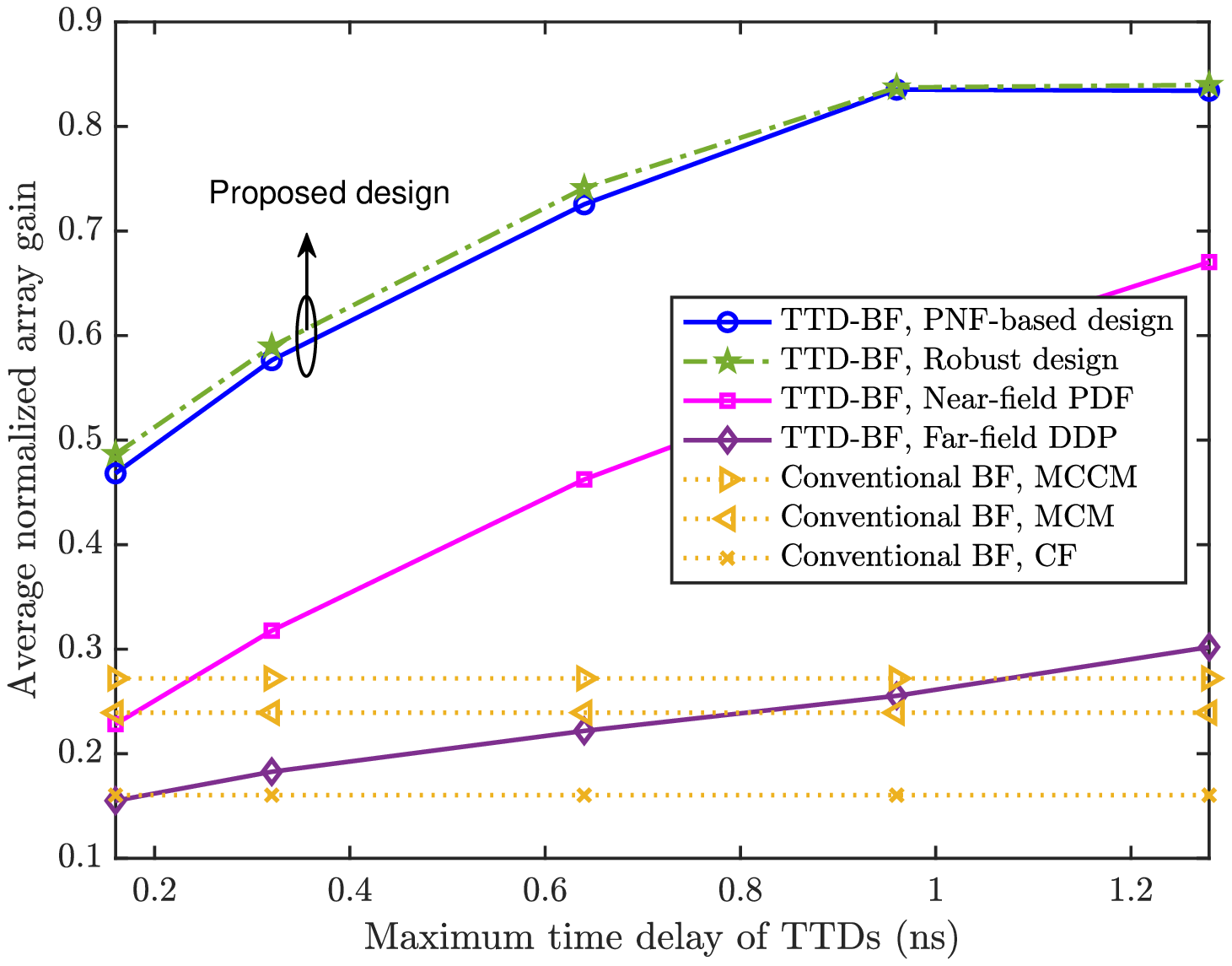}
    \caption{Average array gain versus maximum time delay, $t_{\max}$. }
    \label{fig:array_gain}
    \includegraphics[width=0.4\textwidth]{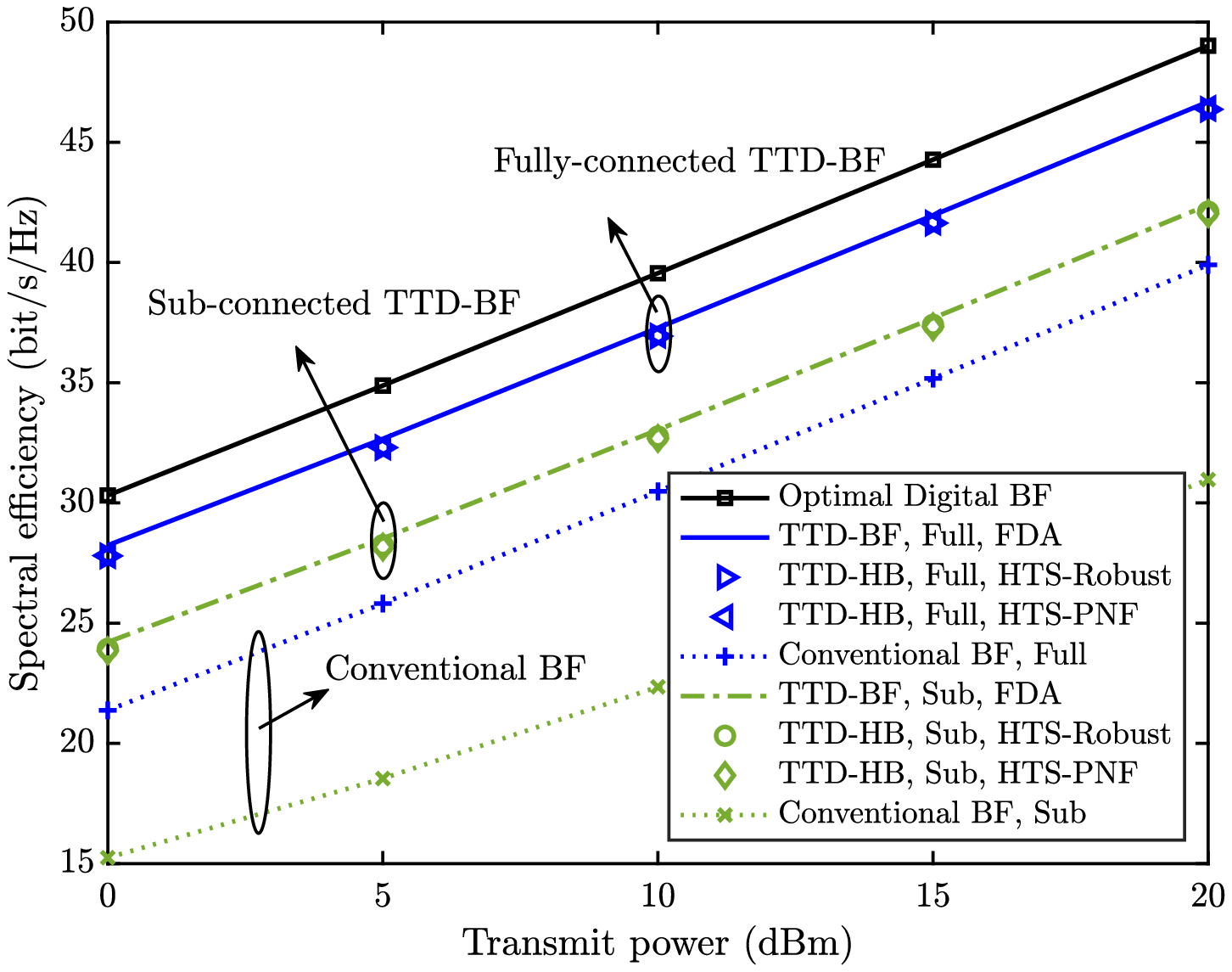}
    \caption{Spectral efficiency versus transmit power, $P_t$.}
    \label{fig:rate_vs_power}
\end{figure}

\vspace{-0.3cm}
\subsection{Spectral and Energy Efficiency Versus Number of TTDs}
In Fig. \ref{fig:number_TTD}, we investigate the effect of the number of TTDs on the spectral efficiency of the considered communication system. The results show that as the number of TTDs increases, the performance of the fully-connected architecture gradually converges towards the optimal digital BF performance. This improvement is attributed to the enhanced mitigation of spatial wideband effects by installing more TTDs. However, this improvement is not evident in the sub-connected architecture. In other words, the sub-connect architecture requires a much smaller number of TTDs to reach its optimal performance due to the weaker spatial wideband effect of the small-aperture sub-arrays connected to each RF chain. This result is consistent with the analysis in \textbf{Remark \ref{remark_1}}. Furthermore, an intriguing phenomenon is that the HTS-PNF approach in the fully-connected architecture experiences a significant performance degradation when the number of TTDs is limited (i.e., $N_T \le 8$). This is because the PNF approximation becomes invalid in this case, see \textbf{Proposition \ref{proposition_1}}. In this case, the HTS-Robust approach should be adopted to guarantee the performance. On the contrary, the performance of the HTS-PNF approach is consistently comparable to that of the FDA approach in the sub-connected architecture. This is also expected since the PNF approximation is only needed to approximate the array response of small sub-arrays, which requires fewer TTDs to ensure high accuracy.

In Fig. \ref{fig:number_TTD_EE}, we further demonstrate the energy efficiency versus the number of TTDs. Energy efficiency is defined as the ratio between spectral efficiency and power consumption. In particular, let $P_{\mathrm{BB}} = 300$ mW, $P_{\mathrm{RF}} = 200$ mW, $P_{\mathrm{PS}} = 30$ mW, and $P_{\mathrm{TTD}} = 100$ mW denote the power consumption of baseband processing, a RF chain, a PS, and a TTD, respectively \cite{dai2022delay}. Then, the power consumption of optimal digital BF architectures is given by $P_{\mathrm{FD}} = P_t + P_{\mathrm{BB}} + N P_{\mathrm{RF}}$. The power consumption of the fully-connected and sub-connected TTD-BF architectures is given by $P_{\mathrm{F \textendash TTD}} = P_t + P_{\mathrm{BB}} + N_{\mathrm{RF}} P_{\mathrm{RF}} + N N_{\mathrm{RF}} P_{\mathrm{PS}} + N N_{\mathrm{T}} P_{\mathrm{TTD}}$ and $P_{\mathrm{S \textendash TTD}} = P_t + P_{\mathrm{BB}} + N_{\mathrm{RF}} P_{\mathrm{RF}} + N P_{\mathrm{PS}} + N N_{\mathrm{T}} P_{\mathrm{TTD}}$, respectively. Similarly, the power consumption of conventional BF architectures can be calculated by excluding the power consumption attributed to TTDs. As can be observed from Fig. \ref{fig:number_TTD_EE}, the sub-connected architecture is much more energy-efficient than the fully-connected architecture. Furthermore, it is evident that the energy efficiency of the fully-connected architecture has peak points at $N_{\mathrm{T}} = 8$ and $16$. This is because, beyond this point, the additional increase in the number of TTDs leads to marginal performance improvements in spectral efficiency, as depicted in Fig. \ref{fig:number_TTD}, while incurring a substantial power consumption. Conversely, the energy efficiency of the sub-connected architecture exhibits a continuous decline from the outset, as the installation of additional TTDs fails to provide significant improvements in spectral efficiency.

\begin{figure}[t!]
    \centering
    \includegraphics[width=0.4\textwidth]{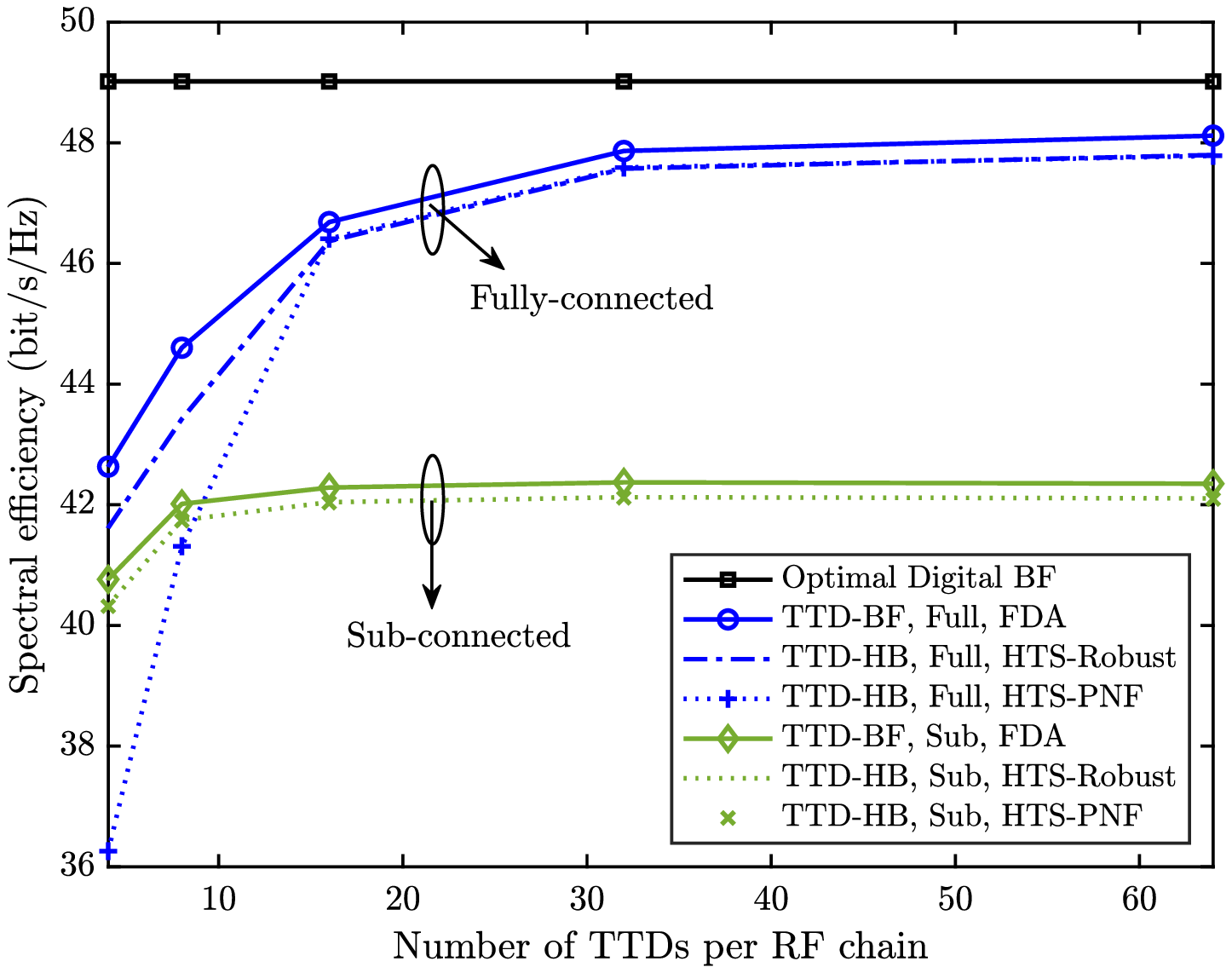}
    \caption{Spectral efficiency versus number of TTDs, $N_{\mathrm{T}}$.}
    \label{fig:number_TTD}
    \includegraphics[width=0.4\textwidth]{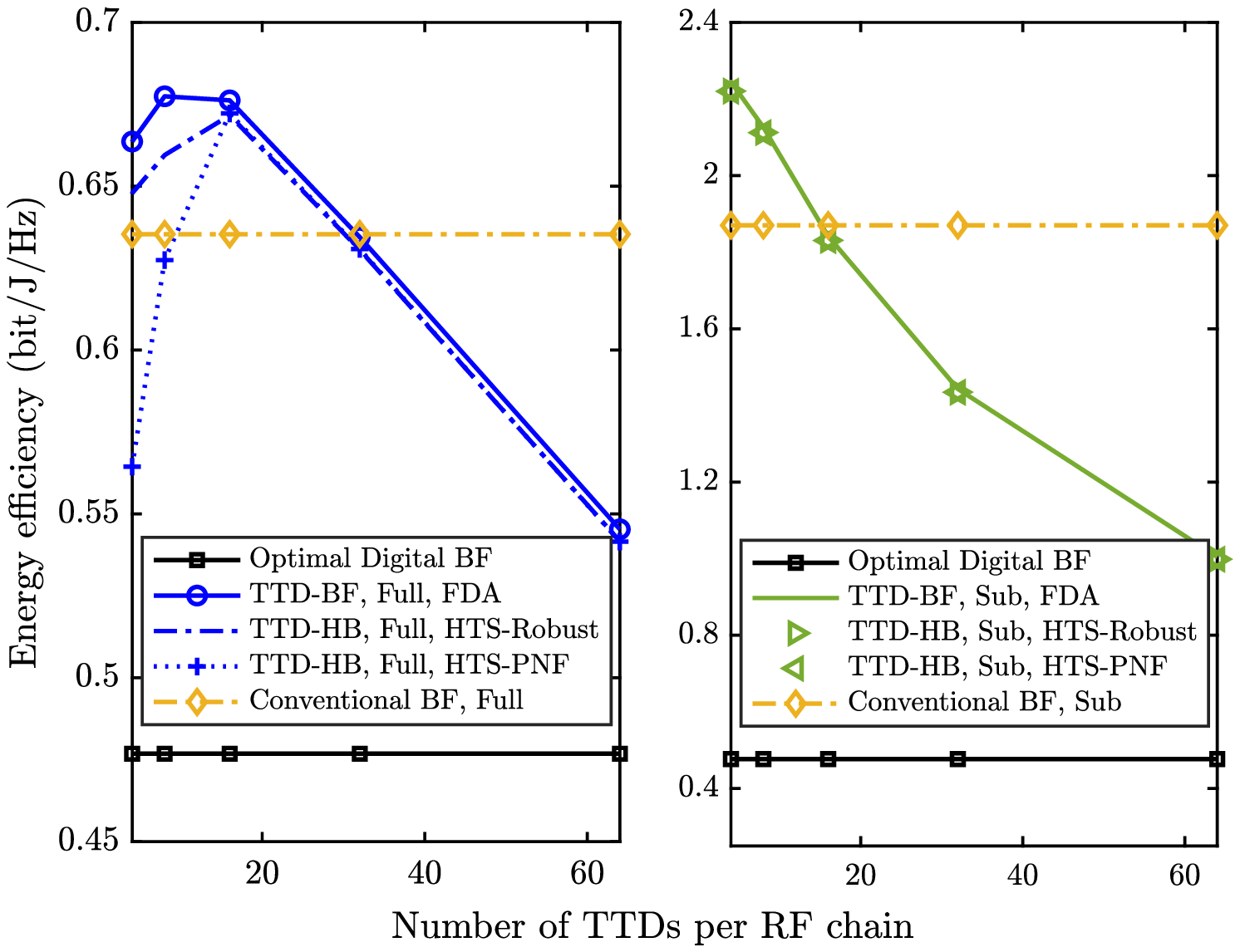}
    \caption{Energy efficiency versus number of TTDs, $N_{\mathrm{T}}$.}
    \label{fig:number_TTD_EE}
\end{figure}

\vspace{-0.3cm}
\subsection{Spectral Efficiency Versus Maximum Delay of TTDs}
In Fig. \ref{fig:max_delay}, the impact of the maximum delay $t_{\max}$ of TTD on the spectral efficiency performance is investigated. The results show that the performance of fully connected architectures is heavily influenced by the value of $t_{\max}$.  Specifically, a minimum of $0.96$ ns of $t_{\max}$ is required to achieve optimal performance. However, for the sub-connected architecture, the minimum $t_{\max}$ to achieve near-optimal performance is less than $0.32$ ns. This phenomenon can be explained as follows. In the sub-connected architecture, the propagation delay difference within each sub-array is significantly smaller than the propagation delay difference of the entire array. Therefore, in order to mitigate the spatial wideband effect in each subarray, the TTD only needs to introduce a small time delay to compensate for the propagation delay difference.

\begin{figure}[t!]
    \centering
    \includegraphics[width=0.4\textwidth]{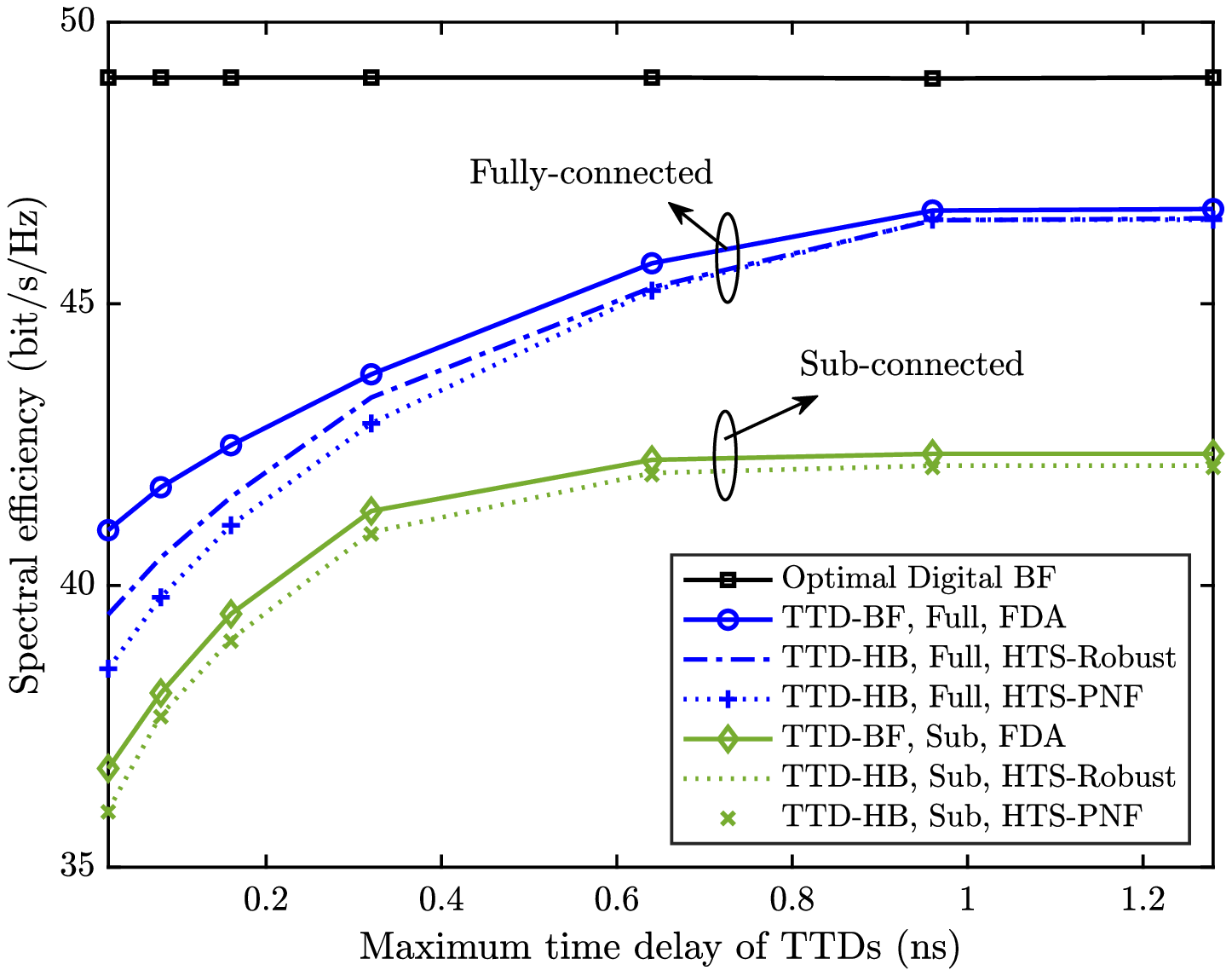}
    \caption{Spectral efficiency versus maximum time delay, $t_{\max}$.}
    \label{fig:max_delay}
    \includegraphics[width=0.4\textwidth]{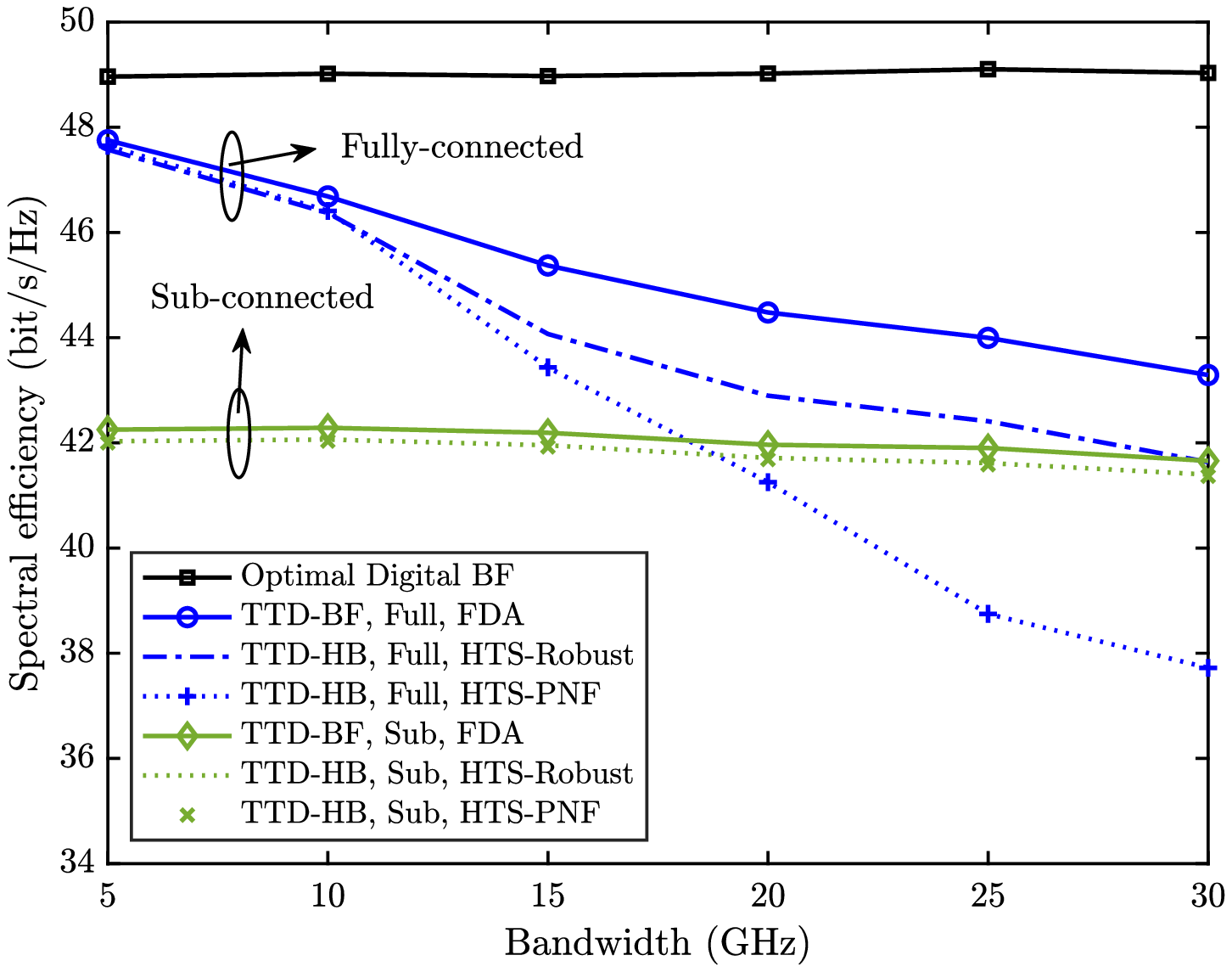}
    \caption{Spectral efficiency versus bandwidth, $B$.}
    \label{fig:bandwidth}
\end{figure} 

\vspace{-0.2cm}
\subsection{Spectral Efficiency Versus Bandwidth}
In Fig. \ref{fig:bandwidth}, we study how the system bandwidth $B$ impacts spectral efficiency. We keep the transmit power $P_t$ adjusted with the bandwidth so that the transmit signal-to-noise ratio $P_t/\sigma_{m,k}^2$ remains consistent. It can be seen that as bandwidth increases, the performance of the fully-connected architecture declines significantly due to the more significant spatial wideband effect, requiring more TTDs to counter it effectively. In contrast, the sub-connected architecture shows negligible performance loss with increasing bandwidth because its smaller sub-array apertures result in a less severe spatial wideband effect. Furthermore, the HTS-PNF approach in fully-connected architectures experiences notable performance loss as bandwidth increases, which is primarily due to the reduced accuracy of the piecewise near-field approximation, as detailed in \textbf{Proposition \ref{proposition_1}}. Compared to the HTS-PNF approach, the HTS-Robust performs more robustly against the increased bandwidth in the fully-connected architectures. Conversely, in sub-connected architectures, the smaller size of the sub-arrays diminishes the impact of increased bandwidths, thereby enhancing the efficiency of the HTS-PNF approach. 

\vspace{-0.2cm}
\section{Conclusion} \label{sec:conclusion}



This study explored beamforming optimization in near-field wideband multi-user communication systems to counteract spatial wideband effects. We introduced a new sub-connected architecture for TTD-based hybrid beamforming and proposed two wideband beamforming methods, FDA and low-complexity HTS, applicable to both fully-connected and sub-connected architectures. Numerical results indicated that while the sub-connected architecture sacrifices some spectral efficiency, it significantly enhances energy efficiency and robustness against increases in system bandwidth and reductions in TTD delays.


\begin{appendices}

\vspace{-0.5cm}
\section{Proof of Proposition \ref{proposition_1}} \label{proposition_1_proof}

\begin{figure}[t!]
    \centering
    \includegraphics[width=0.4\textwidth]{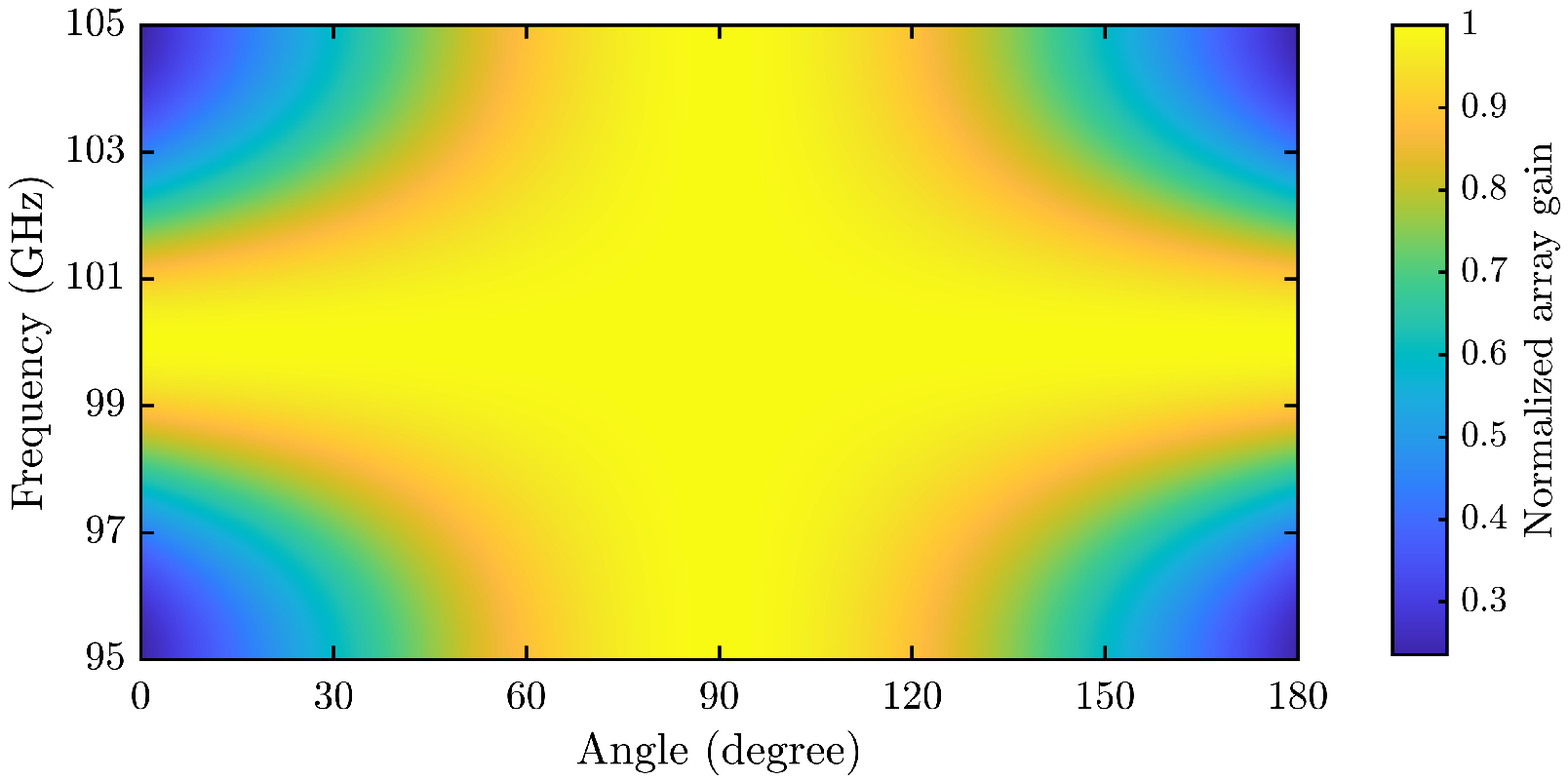}
    \caption{Illustration of the normalized array gain at different frequencies and angles. Here, we set $N = 512$, $N_{\mathrm{T}} = 16$, $f_c = 100$ GHz, $r = 8$ m, and $B = 10$ GHz.}
    \label{fig:approx_error}
\end{figure}

By defining $\tilde{N} = N/N_{\mathrm{T}}$, the normalized array gain can be calculated by 
\vspace{-0.2cm}
\begin{equation}
    G(f_m, \theta, r) = \frac{1}{N} \left|  \sum_{l=1}^{N_{\mathrm{T}}} \sum_{q = 1}^{\tilde{N}} e^{j \frac{2 \pi (f_m - f_c)}{c} (\tilde{\nu}_{l,q} - \nu_l) } \right|.
\end{equation}
Based on numerous simulations, we find that, typically, the minimum value of $G(f_m, \theta, r)$ is achieved when $\theta$ is either 0 or $\pi$ and $f_m$ is either $f_1$ or $f_M$, as illustrated in Fig. \ref{fig:approx_error}, except in cases where the minimum value approaches very close to zero. Given that the threshold $\Delta$ is slightly smaller than one, we focus on the cases of $\theta = 0$ or $\pi$ and $f_m = f_1$ or $f_M$. Let us take the case when $\theta = 0$ and $f_m = f_M$ as an example. In this case, it can be readily obtained that
\vspace{-0.2cm}
\begin{equation}
    \vartheta_l = 0, \quad \tilde{\nu}_{l,q} = \nu_l - \tilde{\chi}_q d, \quad f_M - f_c \approx \frac{B}{2},  \quad  \forall l, q.
    \vspace{-0.2cm}
\end{equation}
Thus, by defining $\eta = d/c = 1/(2 f_c)$, the normalized array when $\theta=0$ at frequency $f_M$ can be simplified as 
\begin{align}
    &G(f_M, 0, r) \approx \frac{1}{N} \left| \sum_{l=1}^{N_{\mathrm{T}}} \sum_{q = 1}^{\tilde{N}} e^{-j \pi \eta B \tilde{\chi}_q } \right| \nonumber \\ 
    &= \left| \frac{N_{\mathrm{T}}}{N} \sum_{q = -\frac{\tilde{N}-1}{2}}^{\frac{\tilde{N}-1}{2}} e^{-j \pi \eta B q} \right|
    = \left| \frac{N_{\mathrm{T}}}{N} \frac{\sin \left( \frac{\pi B N}{4 f_c N_{\mathrm{T}}}  \right)}{\sin \left(\frac{\pi B}{4 f_c} \right)} \right|,
\end{align}
where the last step stems from the definition of the Dirichlet kernel, i.e., $ \sum_{k=-n}^n e^{jkx} = \frac{\sin((n + 1/2)x)}{\sin(x/2)} $. We note that the above formula is unrelated to the distance $r$. Therefore, based on the above analysis, the inequality $\min_{f_m, \theta, r} G(f_m, \theta, r) \ge \Delta$ is approximately achieved if $G(f_M, 0, r) \ge 0$. The cases when $\theta = \pi$ and/or $f_m = f_1$ can be analyzed following the same procedure and have the same conclusion as the considered example. The proof is thus completed.    

\vspace{-0.3cm}
\section{Proof of Proposition \ref{proposition_2}} \label{proposition_2_proof}
For the objective function of problem \eqref{HTS_opt_2}, we get the following result from the triangle inequality:
\begin{align}
    &\left| \sum_{l=1}^{N_{\mathrm{T}}} e^{-j 2\pi f_m \left( \frac{\nu_l(\theta_k, r_k) - r_k}{c} + t_{k,l}  \right)} \right| 
    \le N_{\mathrm{T}},
\end{align}
where the maximum value is achieved when all summation terms have the same arbitrary phase. We set this arbitrary phase to $-2 \pi f_m \zeta_{k}$. The optimal $t_{k,l}$ must satisfy
\begin{align} \label{eq_79}
    & -2\pi f_m \left( \frac{\nu_l(\theta_k, r_k) - r_k}{c} + t_{k,l}  \right) = -2 \pi f_m \zeta_k \\
    \Leftrightarrow \quad & t_{k,l} = \zeta_k - \frac{\nu_l(\theta_k, r_k) - r_k}{c}.
\end{align}  
To ensure that $t_{k,l} \ge 0, \forall l$, the value of $\zeta_k$ need to be set to 
\begin{equation} \label{eq_80}
    \zeta_k = \frac{\nu_{\max,k} - r_k}{c}, 
\end{equation}  
where $\nu_{\max,k} = \max \{\nu_1(\theta_k, r_k),\dots, \nu_{N_{\mathrm{T}}}(\theta_k, r_k) \}$. Substituting \eqref{eq_80} into \eqref{eq_79} yields 
\begin{equation} \label{optimal_closed_form_t}
    t_{k,l} = \frac{\nu_{\max,k} - \nu_l(\theta_k, r_k)}{c}.
\end{equation} 
According to the geometric relationship illustrated in Fig. \ref{fig:channel_approx}, the largest possible value of $t_{k,l}$ given in the above solution is obtained when $\theta_k = 0$ or $\pi$, which is given by (c.f. \eqref{eq_50}) 
\begin{align}
    t_{k,l} \le &\frac{\nu_1(\pi, r_k) - \nu_{N_{\mathrm{T}}}(\pi, r_k)}{c}
    = \frac{N (N_{\mathrm{T}} - 1)d}{N_{\mathrm{T}} c}.
\end{align}
Therefore, if $t_{\max} \ge \frac{N (N_{\mathrm{T}} - 1)d}{N_{\mathrm{T}} c}$, the optimal solution in \eqref{optimal_closed_form_t} can always be achieved. The proof is thus completed.



\end{appendices}

\bibliographystyle{IEEEtran}
\bibliography{reference/mybib}

\begin{thebibliography}{10}
\providecommand{\url}[1]{#1}
\csname url@samestyle\endcsname
\providecommand{\newblock}{\relax}
\providecommand{\bibinfo}[2]{#2}
\providecommand{\BIBentrySTDinterwordspacing}{\spaceskip=0pt\relax}
\providecommand{\BIBentryALTinterwordstretchfactor}{4}
\providecommand{\BIBentryALTinterwordspacing}{\spaceskip=\fontdimen2\font plus
\BIBentryALTinterwordstretchfactor\fontdimen3\font minus \fontdimen4\font\relax}
\providecommand{\BIBforeignlanguage}[2]{{%
\expandafter\ifx\csname l@#1\endcsname\relax
\typeout{** WARNING: IEEEtran.bst: No hyphenation pattern has been}%
\typeout{** loaded for the language `#1'. Using the pattern for}%
\typeout{** the default language instead.}%
\else
\language=\csname l@#1\endcsname
\fi
#2}}
\providecommand{\BIBdecl}{\relax}
\BIBdecl

\bibitem{conference_version}
{Z. Wang \emph{et al.}}, ``Near-field wideband beamfocusing optimization: A heuristic two-stage approach,'' in \emph{Proc. {IEEE} Global Commun. Conf. ({GLOBECOM})}, Kuala Lumpur, Malaysia, Dec. 2023, pp. 783--788.

\bibitem{zhang20196g}
{Z. Zhang \emph{et al.}}, ``6{G} wireless networks: Vision, requirements, architecture, and key technologies,'' \emph{IEEE Veh. Technol. Mag.}, vol.~14, no.~3, pp. 28--41, Sep. 2019.

\bibitem{kraus2002antennas}
J.~D. Kraus and R.~J. Marhefka, \emph{Antennas for all applications}.\hskip 1em plus 0.5em minus 0.4em\relax New York, NY, USA: McGraw-Hill, 2002.

\bibitem{liu2023near}
{Y. Liu \emph{et al.}}, ``Near-field communications: A tutorial review,'' \emph{IEEE Open J. Commun. Soc.}, vol.~4, pp. 1999--2049, Aug. 2023.

\bibitem{heath2016overview}
{R. W. Heat \emph{et al.}}, ``An overview of signal processing techniques for millimeter wave {MIMO} systems,'' \emph{{IEEE} J. Sel. Topics Signal Process.}, vol.~10, no.~3, pp. 436--453, Apr. 2016.

\bibitem{zhang2022beam}
{H. Zhang \emph{et al.}}, ``Beam focusing for near-field multiuser {MIMO} communications,'' \emph{{IEEE} Trans. Wireless Commun.}, vol.~21, no.~9, pp. 7476--7490, Sep. 2022.

\bibitem{yu2016alternating}
{X. Yu \emph{et al.}}, ``Alternating minimization algorithms for hybrid precoding in millimeter wave {MIMO} systems,'' \emph{{IEEE} J. Sel. Topics Signal Process.}, vol.~10, no.~3, pp. 485--500, Apr. 2016.

\bibitem{sohrabi2016hybrid}
F.~Sohrabi and W.~Yu, ``Hybrid digital and analog beamforming design for large-scale antenna arrays,'' \emph{{IEEE} J. Sel. Topics Signal Process.}, vol.~10, no.~3, pp. 501--513, Apr. 2016.

\bibitem{wang2018spatial}
B.~Wang, F.~Gao, S.~Jin, H.~Lin, and G.~Y. Li, ``Spatial-and frequency-wideband effects in millimeter-wave massive {MIMO} systems,'' \emph{{IEEE} Trans. Signal Process.}, vol.~66, no.~13, pp. 3393--3406, Jul. 2018.

\bibitem{lu2021near}
H.~Lu and Y.~Zeng, ``Near-field modeling and performance analysis for multi-user extremely large-scale {MIMO} communication,'' \emph{{IEEE} Commun. Lett.}, vol.~26, no.~2, pp. 277--281, Feb. 2022.

\bibitem{10339299}
{Y. Li \emph{et al.}}, ``Near-field beamforming optimization for holographic {XL-MIMO} multiuser systems,'' \emph{{IEEE} Trans. Commun.}, vol.~72, no.~4, pp. 2309--2323, Apr. 2024.

\bibitem{bacci2023spherical}
{G. Bacci \emph{et al.}}, ``Spherical wavefronts improve {MU-MIMO} spectral efficiency when using electrically large arrays,'' \emph{{IEEE} Wireless Commun. Lett.}, vol.~12, no.~7, pp. 1219--1223, Jul. 2023.

\bibitem{chen2020hybrid}
{Y. Chen \emph{et al.}}, ``Hybrid precoding for wideband millimeter wave {MIMO} systems in the face of beam squint,'' \emph{{IEEE} Trans. Wireless Commun.}, vol.~20, no.~3, pp. 1847--1860, Mar. 2021.

\bibitem{gao2021wideband}
F.~Gao, B.~Wang, C.~Xing, J.~An, and G.~Y. Li, ``Wideband beamforming for hybrid massive {MIMO} terahertz communications,'' \emph{{IEEE} J. Sel. Areas Commun.}, vol.~39, no.~6, pp. 1725--1740, Jun. 2021.

\bibitem{6531062}
M.~Longbrake, ``True time-delay beamsteering for radar,'' in \emph{Proc. IEEE Nat. Aerosp. Electron. Conf. (NAECON)}, Jul. 2012, pp. 246--249.

\bibitem{dai2022delay}
L.~Dai, J.~Tan, Z.~Chen, and H.~V. Poor, ``Delay-phase precoding for wideband {THz} massive {MIMO},'' \emph{{IEEE} Trans. Wireless Commun.}, vol.~21, no.~9, pp. 7271--7286, Sep. 2022.

\bibitem{nguyen2022joint}
D.~Q. Nguyen and T.~Kim, ``Joint delay and phase precoding under true-time delay constraints for {THz} massive {MIMO},'' in \emph{Proc. IEEE Int. Conf. Commun. (ICC)}, May 2022, pp. 3496--3501.

\bibitem{dovelos2021channel}
{K. Dovelos \emph{et al.}}, ``Channel estimation and hybrid combining for wideband terahertz massive {MIMO} systems,'' \emph{{IEEE} J. Sel. Areas Commun.}, vol.~39, no.~6, pp. 1604--1620, Jun. 2021.

\bibitem{yan2022energy}
{L. Yan \emph{et al.}}, ``Energy-efficient dynamic-subarray with fixed true-time-delay design for terahertz wideband hybrid beamforming,'' \emph{{IEEE} J. Sel. Areas Commun.}, vol.~40, no.~10, pp. 2840--2854, Oct. 2022.

\bibitem{myers2021infocus}
N.~J. Myers and R.~W. Heath, ``{InFocus}: A spatial coding technique to mitigate misfocus in near-field {LoS} beamforming,'' \emph{{IEEE} Trans. Wireless Commun.}, vol.~21, no.~4, pp. 2193--2209, Apr. 2022.

\bibitem{cui2021near}
M.~Cui and L.~Dai, ``Near-field wideband beamforming for extremely large antenna arrays,'' \emph{arXiv preprint arXiv:2109.10054}, 2021.

\bibitem{zhang2023deep}
Y.~Zhang and A.~Alkhateeb, ``Deep learning of near field beam focusing in terahertz wideband massive {MIMO} systems,'' \emph{{IEEE} Wireless Commun. Lett.}, vol.~12, no.~3, pp. 535--539, Mar. 2023.

\bibitem{zhang2022fast}
Y.~Zhang, X.~Wu, and C.~You, ``Fast near-field beam training for extremely large-scale array,'' \emph{{IEEE} Wireless Commun. Lett.}, vol.~11, no.~12, pp. 2625--2629, Dec. 2022.

\bibitem{wang2024performance}
Z.~Wang, X.~Mu, and Y.~Liu, ``Performance analysis of wideband near-field sensing ({NISE}),'' \emph{arXiv preprint arXiv:2404.05076}, 2024.

\bibitem{jornet2011channel}
J.~M. Jornet and I.~F. Akyildiz, ``Channel modeling and capacity analysis for electromagnetic wireless nanonetworks in the terahertz band,'' \emph{{IEEE} Trans. Wireless Commun.}, vol.~10, no.~10, pp. 3211--3221, Oct. 2011.

\bibitem{piesiewicz2007scattering}
{R. Piesiewicz \emph{et al.}}, ``Scattering analysis for the modeling of {THz} communication systems,'' \emph{IEEE Trans. Antennas Propag.}, vol.~55, no.~11, pp. 3002--3009, Nov. 2007.

\bibitem{priebe2013stochastic}
S.~Priebe and T.~Kurner, ``Stochastic modeling of {THz} indoor radio channels,'' \emph{{IEEE} Trans. Wireless Commun.}, vol.~12, no.~9, pp. 4445--4455, Sep. 2013.

\bibitem{5648370}
A.~M. Abbosh, ``Broadband fixed phase shifters,'' \emph{IEEE Microw. Wirel. Compon. Lett.}, vol.~21, no.~1, pp. 22--24, Jan. 2011.

\bibitem{4623722}
H.~Hashemi, T.-s. Chu, and J.~Roderick, ``Integrated true-time-delay-based ultra-wideband array processing,'' \emph{{IEEE} Commun. Mag.}, vol.~46, no.~9, pp. 162--172, Sep. 2008.

\bibitem{8248806}
{J.-C. Jeong \emph{et al.}}, ``A 6–18-{GHz} {GaAs} multifunction chip with 8-bit true time delay and 7-bit amplitude control,'' \emph{IEEE Trans. Microw. Theory Techn.}, vol.~66, no.~5, pp. 2220--2230, May 2018.

\bibitem{shi2020penalty}
Q.~Shi and M.~Hong, ``Penalty dual decomposition method for nonsmooth nonconvex optimization—part {I}: Algorithms and convergence analysis,'' \emph{{IEEE} Trans. Signal Process.}, vol.~68, pp. 4108--4122, Jun. 2020.

\bibitem{6832894}
{E. Björnson \emph{et al.}}, ``Optimal multiuser transmit beamforming: A difficult problem with a simple solution structure [lecture notes],'' \emph{IEEE Signal Process. Mag.}, vol.~31, no.~4, pp. 142--148, Jul. 2014.

\bibitem{shen2018fractional}
K.~Shen and W.~Yu, ``Fractional programming for communication systems—part {I}: Power control and beamforming,'' \emph{{IEEE} Trans. Signal Process.}, vol.~66, no.~10, pp. 2616--2630, May 2018.

\bibitem{shen2018fractional2}
------, ``Fractional programming for communication systems—part {II}: Uplink scheduling via matching,'' \emph{{IEEE} Trans. Signal Process.}, vol.~66, no.~10, pp. 2631--2644, May 2018.

\bibitem{hu2020iterative}
{Q. Hu \emph{et al.}}, ``Iterative algorithm induced deep-unfolding neural networks: Precoding design for multiuser {MIMO} systems,'' \emph{{IEEE} Trans. Wireless Commun.}, vol.~20, no.~2, pp. 1394--1410, Feb. 2021.

\bibitem{shlezinger2022model}
N.~Shlezinger, Y.~C. Eldar, and S.~P. Boyd, ``Model-based deep learning: On the intersection of deep learning and optimization,'' \emph{IEEE Access}, vol.~10, pp. 115\,384--115\,398, Nov. 2022.

\bibitem{chen2023beam}
K.~Chen, C.~Qi, C.-X. Wang, and G.~Y. Li, ``Beam training and tracking for extremely large-scale {MIMO} communications,'' \emph{{IEEE} Trans. Wireless Commun.}, early access, Oct. 2023. doi: 10.1109/TWC.2023.3324176.

\bibitem{tan2021wideband}
J.~Tan and L.~Dai, ``Wideband beam tracking in thz massive {MIMO} systems,'' \emph{{IEEE} J. Sel. Areas Commun.}, vol.~39, no.~6, pp. 1693--1710, Jun. 2021.

\bibitem{alkhateeb2015limited}
A.~Alkhateeb, G.~Leus, and R.~W. Heath, ``Limited feedback hybrid precoding for multi-user millimeter wave systems,'' \emph{{IEEE} Trans. Wireless Commun.}, vol.~14, no.~11, pp. 6481--6494, Nov. 2015.

\bibitem{sun2016majorization}
Y.~Sun, P.~Babu, and D.~P. Palomar, ``Majorization-minimization algorithms in signal processing, communications, and machine learning,'' \emph{{IEEE} Trans. Signal Process.}, vol.~65, no.~3, pp. 794--816, Feb. 2017.

\bibitem{christensen2008weighted}
{S. S. Christensen \emph{et al.}}, ``Weighted sum-rate maximization using weighted {MMSE} for {MIMO-BC} beamforming design,'' \emph{{IEEE} Trans. Wireless Commun.}, vol.~7, no.~12, pp. 4792--4799, Dec. 2008.

\end{thebibliography}

\end{document}